\documentclass[12pt,reqno,a4paper]{amsart}
\usepackage{amsmath,amssymb,dsfont,graphicx,cite,latexsym,
epsf,cancel,epic,eepic}
\usepackage[colorlinks=false]{hyperref}
\usepackage{mathpazo}
\usepackage{tikz}
\usepackage{subfig}

\usetikzlibrary{arrows.meta}
\usetikzlibrary{arrows}
\newtheorem{theorem}{Theorem}

\newtheorem{lemma}{Lemma}

\newtheorem{definition}{Definition}

\def\beq{\begin{equation}}
\def\eeq{\end{equation}}
\def\bea{\begin{eqnarray}}
\def\eea{\end{eqnarray}}
\newcommand{\tetr}[4]{
\raisebox{-7mm}{
\begin{tikzpicture}[scale=.3]
	\draw (-1,0) node[left] {$#1$} -- 
	                 (0,-0.8) node[below] {$#3$} --
			(1,0) node[right] {$#2$} --
			(0,1) node[above] {$#4$} -- cycle;
      \draw[dashed] (-1,0) -- (1,0)	;
      \draw (0,1) -- (0,-0.8);
\end{tikzpicture}
}
}
\newcommand{\tri}[3]{
\raisebox{-7mm}{
\begin{tikzpicture}[scale=.25]
	\draw (-1,0) node[below left] {$#1$} -- 
			(1,0) node[below right] {$#2$} --
			(0,1.732) node[above] {$#3$} -- cycle;
\end{tikzpicture}
}
}
\makeatletter
\expandafter\let\expandafter
\reset@font\csname reset@font\endcsname
\def\subeqnarray{\arraycolsep1pt
   \def\@eqnnum\stepcounter##1{\stepcounter{subequation}
       {\reset@font\rm(\theequation\alph{subequation})}}
\jot5mm     \eqnarray}

\makeatother
\newcommand{\bbZ}{{\mathbb Z}}
\newcommand{\bbR}{{\mathbb R}}
\newcommand{\bbC}{{\mathbb C}}

\newcommand{\cL}{{\mathcal L}}

\def\epsilon{\varepsilon}

\def\bbC{\mathbb C}

\def\bbZ{\mathbb Z}

\newbox\meibox
\def\placeunder#1#2#3#4{\setbox\meibox%
\vbox{\hbox{\hskip#4$\hphantom{#2}$}\hbox{$\hphantom{#1}$}}%
\vtop{\baselineskip=0pt\lineskiplimit=\baselineskip%
\lineskip=#3\hbox to \wd\meibox{\hfil\hskip#4$#2$\hfil}%
\hbox to \wd\meibox{\hfil$#1$\hfil}}}
\def\intprod{\mathbin{\hbox to 6pt{%
                 \vrule height0.4pt width5pt depth0pt
                 \kern-.4pt
                 \vrule height6pt width0.4pt depth0pt\hss}}}

\begin{document}
\title[Linear integrable systems on quad-graphs]
{Linear integrable systems on quad-graphs}

%
\author{Alexander I. Bobenko \and Yuri B. Suris}

\thanks{E-mail: {\tt  bobenko@math.tu-berlin.de, suris@math.tu-berlin.de}}

\maketitle

\begin{center}
{\footnotesize{
Institut f\"ur Mathematik,
Technische Universit\"at Berlin, Str. des 17. Juni 136,
10623 Berlin, Germany
}}
\end{center}

\begin{abstract}

In the first part of the paper, we classify linear integrable (multi-dimensionally consistent) quad-equations on bipartite isoradial quad-graphs in $\bbC$, enjoying natural symmetries and the property that the restriction of their solutions to the black vertices satisfies a Laplace type equation. The classification reduces to solving a functional equation. Under certain restriction, we give a complete solution of the functional equation, which is expressed in terms of elliptic functions. We find two real analytic reductions, corresponding to the cases when the underlying complex torus is of a rectangular type or of a rhombic type. The solution corresponding to the rectangular type was previously found by Boutillier, de Tili\`ere and Raschel. Using the multi-dimensional consistency, we construct the discrete exponential function, which serves as a basis of solutions of the quad-equation. 

In the second part of the paper, we focus on the integrability of discrete linear variational problems. We consider discrete pluri-harmonic functions, corresponding to a discrete 2-form with a quadratic dependence on the fields at black vertices only. In an important particular case, we show that the problem reduces to a two-field generalization of the classical star-triangle map. We prove the integrability of this novel 3D system by showing its multi-dimensional consistency. The Laplacians from the first part come as a special solution of the two-field star-triangle map. 
\end{abstract}

\section{Introduction}
\label{sect: Introduction}

Integrable systems are usually considered as belonging to the nonlinear mathematical physics.
They appear as compatibility conditions of overdetermined systems of (auxiliary) linear equations.
It turns out that linear systems with similar structures are equally if not even more important than the nonlinear ones (compare, for example, characterizations of Jacobi varieties in terms of solutions of linear integrable systems \cite{Krichever} and in terms of nonlinear integrable equations \cite{Shiota}).

This paper is devoted to a classification of discrete linear integrable equations. We consider linear quad equations. These are equations relating fields at four vertices of faces of quad-graphs. The latter are cell-decompositions of surfaces with quadrilateral faces. A standard example is the square grid ${\mathbb Z}^2$. Integrability of quad equations is understood as the multi-dimensional consistency \cite{BS_quad, BS_DDG_book}. The theory we are developing here can be seen as a generalization of discrete complex analysis on isoradial graphs \cite{Kenyon, Chelkak_Smirnov}. We are guided by the structures of discrete complex analysis with multi-dimensionally consistent discrete Cauchy-Riemann equations \cite{BMS}. We are motivated by new examples of massive discrete Dirac and Laplace operators found in \cite{BTR, dT}. Our goal is to place these examples in a comprehensive picture.

Let us describe the setup. The underlying combinatorial structure is a bipartite quad-graph $G$, rhombically embedded in $\mathbb C$, with black and white vertices $V(G)=B\cup W$. Solutions of the quad equations restricted to black (white) vertices are required to satisfy Laplace type equations, which are equations on black (resp. white) stars, with the edges being the diagonals of the original quads. This resembles the relation between discrete holomorphic and discrete harmonic functions. The main result of this paper is a classification of linear quad-equations of this type (see Theorems \ref{th 1}, \ref{th 2}).

It is worth to mention that in the process of classification remarkable new structures appeared. A central role in our analysis is played by the functional equation
\begin{align*}
f(\alpha,\beta)f(\beta,\gamma)f(\gamma,\alpha)+f(\beta,\alpha)f(\alpha,\delta)f(\delta,\beta)&+f(\gamma,\beta)f(\beta,\delta)f(\delta,\gamma)\nonumber \\
&+f(\alpha,\gamma)f(\gamma,\delta)f(\delta,\alpha)=0.
\end{align*}
Combinatorially, this equation can be treated as the closeness condition 
$$
d\omega\bigg( \tetr{\alpha}{\beta}{\gamma}{\delta} \bigg)=0,
$$
on a tetrahedron with the vertices carrying the fields $\alpha$, $\beta$, $\gamma$, $\delta$, for the discrete 2-form $\omega$ defined as
$$
\omega\bigg( \tri{\alpha}{\beta}{\gamma} \bigg)=f(\alpha,\beta)f(\beta,\gamma)f(\gamma,\alpha).
$$
We give a complete solution of this equation for the case $f(\alpha,\beta)=f(\alpha-\beta)$, when the function $f$ depends on the difference of its arguments. It is given in terms of elliptic theta-functions. When the underlying complex torus is of rectangular type, we recover the Z-Dirac equation found in \cite{dT}. Another real valued solution comes from complex tori of rhombic type. Restricted to black vertices, solutions of quad-equations result in discrete harmonic functions for the massive Laplace operator found in \cite{BTR}.

Using the multi-dimensional consistency, we construct the discrete exponential function on $V(G)$. It can be treated as (generalized) discrete holomorphic and naturally extends the discrete exponential function on black vertices $B$ introduced in \cite{BTR}, which is (generalized) discrete harmonic.
 
In the second part of the paper we focus on the integrability of discrete linear variational problems. A counterpart of the multi-dimensional consistency in the variational context is the pluri-Lagrangian structure. In the linear case the corresponding systems are called discrete pluri-harmonic \cite{BS_pluriharmonic}. The Lagrangian is a discrete 2-form on quadrilateral faces. In this paper we consider the case when this 2-form depends on the fields at black vertices only, see Fig. ~\ref{fig: two-cells}. This is equivalent to consider a general quadratic form on black diagonals of the quads:
$
{\mathcal L}=\frac{1}{2}(bx^2+cx_{ij}^2)-ax x_{ij}.
$ 
In the case $b=c$ we classify all pluri-harmonic systems. They are described by a two field star-triangle map for the coefficients $a$ and $c$. We prove the integrability of this 3D system by showing its multi-dimensional consistency.
As a special solution of the star-triangle map, one recovers the massive Laplacians from \cite{BTR} considered in the first part of the paper. 

In the Appendix we consider the pluri-harmonic systems satisfying $bc=a^2$, which also includes massive Laplacians. It is shown that they are gauge equivalent to the case $a=b=c$ governed by the standard star-triangle equation.

\section{Classification of integrable linear quad-equations}
\label{sect: classification}

We aim at a classification of the most general linear integrable quad-equations on isoradial bipartite quad-graphs, such that their solutions, being restricted to the black (or white) vertices, satisfy Laplace type equations.  Let $G$ be a quad-graph, rhombically embedded in $\bbC$, with the vertices $V=B\cup W$, $B$ the set of black vertices, $W$ the set of white vertices. All edges are directed from black to white and carry labels which are the corresponding complex numbers $\exp(i\alpha)$ or just their logarithms $\alpha$. Note that opposite edges of every quad carry the labels $\alpha$ and $\alpha+\pi$, see Figure \ref{fig: quad-equation}.

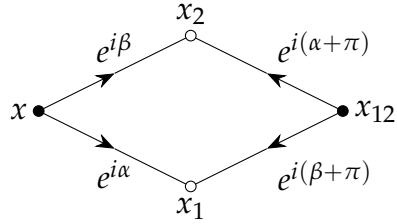
\begin{figure}[htbp]
\centering
\begin{tikzpicture}[scale=1]
\draw[-{Stealth[scale=1.5]}] (0,0) -- (1,.5) node[above] {$e^{i\beta}$}; 
\draw (1,.5) -- (2,1);
\draw[-{Stealth[scale=1.5]}] (0,0) -- (1,-.5) node[below] {$e^{i\alpha}$}; 
\draw (1,-.5) -- (2,-1);
\draw[-{Stealth[scale=1.5]}] (4,0) -- (3,.5) node[above right] {$e^{i(\alpha+\pi)}$}; 
\draw (3,.5) -- (2,1);
\draw[-{Stealth[scale=1.5]}] (4,0) -- (3,-.5) node[below right] {$e^{i(\beta+\pi)}$}; 
\draw (3,-.5) -- (2,-1);

\draw[fill=black] (0,0) circle(2pt) node[left] {$x$};
\draw[fill=white]  (2,1) circle(2pt) node[above] {$x_2$};
\draw[fill=white]  (2,-1) circle(2pt) node[below] {$x_1$};
\draw[fill=black] (4,0) circle(2pt) node[right] {$x_{12}$};
\end{tikzpicture}
\caption[quad-equation]{A quad-equation on a bipartite quad-graph with directed edges}
\label{fig: quad-equation}
\end{figure}

We consider quad-equations of the following (three-leg) form:
\begin{equation}\label{quad eq}
f(\alpha,\beta) x_{12}-g(\alpha,\beta)x_0=i\Big(h(\beta)x_2-h(\alpha)x_1\Big).
\end{equation}
Ansatz \eqref{quad eq} ensures that, upon summation over quads $(x_0,x_k,x_{k,k+1},x_{k+1})$ adjacent to a given (black) vertex $x_0$ (with the labels $\alpha_k$, $\alpha_{k+1}$, say), one ends up with a Laplace-type equation
\begin{equation}\label{Laplace eq}
\sum_{k=1}^m f(\alpha_k,\alpha_{k+1})x_{k,k+1}-\Big(\sum_{k=1}^m g(\alpha_k,\alpha_{k+1})\Big)x_0=0,
\end{equation}
see Figure \ref{Fig: graph and black star}. Actually, the same will be true for vertex stars of white vertices, as well, due to the symmetry conditions which will be imposed on \eqref{quad eq}.

\begin{figure}[htbp]
   \centering
   \subfloat[Rhombic faces of $G$ adjacent to $x_0$]{\label{Fig: star in a graph}
    \begin{tikzpicture}[scale=1.5,inner sep=2]  
      \node (x1) at (1,0) [circle,fill,label=-90:$x_{56}$] {};
      \node (x11) at (3,0) [circle,fill,label=-90:$x_{61}$] {};
      \node (x2) at (0,1.732) [circle,fill,label=135:$x_{45}$] {};
      \node (x12) at (2,1.732) [circle,fill,label=0:$\;\;x_0$] {};
      \node (x112) at (4,1.732) [circle,fill,label=45:$x_{12}$] {};
      \node (x22) at (1,3.464) [circle,fill,label=90:$x_{34}$] {};
      \node (x122) at (3,3.464) [circle,fill,label=90:$x_{23}$] {}; 
      \node (y5) at (1,1.155) [circle,draw,label=-15:$x_5$]{};
      \node (y4) at (1,2.309) [circle,draw,label=-90:$x_4$]{};
       \node (y3) at (2,2.887) [circle,draw,label=90:$x_3$]{};
       \node (y2) at (3,2.310) [circle,draw,label=30:$x_2$]{};
        \node (y1) at (3, 1.155) [circle,draw,label=-15:$x_1$]{};
        \node (y6) at (2,0.577) [circle,draw,label=135:$x_6$]{};
      \draw [very thick] (x12) to (y5) to (x2) to (y4) to (x12);
      \draw [very thick] (y4) to (x22) to (y3) to (x12);
      \draw [very thick] (y3) to (x122) to (y2) to (x12);
      \draw [very thick] (y2) to (x112) to (y1) to (x12);
      \draw [very thick] (y1) to (x11) to (y6) to (x12);
      \draw [very thick] (y6) to (x1) to (y5);
   \end{tikzpicture}}
   \qquad
   \subfloat[Star of $x_0$ in the black graph, supporting the Laplace type equation]{\label{Fig: dual face}
   \begin{tikzpicture}[scale=1.5,inner sep=2]  
      \node (x1) at (1,0) [circle,fill,label=-90:$x_{56}$] {};
      \node (x11) at (3,0) [circle,fill,label=-90:$x_{61}$] {};
      \node (x2) at (0,1.732) [circle,fill,label=135:$x_{45}$] {};
      \node (x12) at (2,1.732) [circle,fill,label=15:$\;x_0$] {};
      \node (x112) at (4,1.732) [circle,fill,label=45:$x_{12}$] {};
      \node (x22) at (1,3.464) [circle,fill,label=90:$x_{34}$] {};
      \node (x122) at (3,3.464) [circle,fill,label=90:$x_{23}$] {}; 
      \draw [very thick] (x12) to (x1);
       \draw [very thick] (x12) to (x11);
      \draw [very thick]  (x12) to (x2);
      \draw [very thick] (x12) to (x112);
      \draw [very thick] (x12) to (x122);
      \draw [very thick] (x12) to (x22);
   \end{tikzpicture}}
  \caption{}
   \label{Fig: graph and black star}
\end{figure}

These conditions are:
\begin{itemize}
\item Functions $f(\alpha,\beta)$, $g(\alpha,\beta)$, $h(\alpha)$ are $2\pi$-periodic with respect to each of their arguments.
\item Invariance with respect to the change of orientation of the quad, $x_1\leftrightarrow x_2$, $\alpha\leftrightarrow\beta$:
\begin{equation}\label{change orient}
f(\alpha,\beta)=-f(\beta,\alpha), \quad g(\alpha,\beta)=-g(\beta,\alpha).
\end{equation} 
\item Invariance under centering the quad at any of the four vertices (see \eqref{quad eq at x12}, \eqref{quad eq at x1}, \eqref{quad eq at x2} for quad-equation \eqref{quad eq} centered at other three vertices of the quad).
\end{itemize}
\begin{lemma}
Quad-equation \eqref{quad eq} with coefficients satisfying \eqref{change orient} is symmetric with respect to centering the quad at any of its four vertices if and only if the following conditions are satisfied:
\begin{align}
& h(\alpha)h(\alpha+\pi)=c, \label{h sym} \\
& f(\alpha,\beta)f(\alpha+\pi,\beta)=-c, \label{f sym}\\
& g(\alpha,\beta)=c^{-1}f(\alpha,\beta)h(\alpha)h(\beta), \label{g thru fh}
\end{align}
with some constant $c\neq 0$. From \eqref{change orient}, \eqref{f sym} there follows also
\begin{equation}\label{f per}
f(\alpha,\beta)=f(\alpha+\pi,\beta+\pi).
\end{equation}
\end{lemma}
\begin{proof}
Equation centered at $x_{12}$:
\begin{equation}\label{quad eq at x12}
f(\alpha+\pi,\beta+\pi) x_{0}-g(\alpha+\pi,\beta+\pi)x_{12}=i\Big(h(\beta+\pi)x_1-h(\alpha+\pi)x_2\Big).
\end{equation}
Equation centered at $x_{1}$:
\begin{equation}\label{quad eq at x1}
f(\beta+\pi,\alpha) x_{2}-g(\beta+\pi,\alpha)x_{1}=i\Big(h(\alpha)x_0-h(\beta+\pi)x_{12}\Big).
\end{equation}
Equation centered at $x_{2}$:
\begin{equation}\label{quad eq at x2}
f(\alpha+\pi,\beta) x_{1}-g(\alpha+\pi,\beta)x_{2}=i\Big(h(\beta)x_0-h(\alpha+\pi)x_{12}\Big).
\end{equation}
Thus, the symmetry requirement reads:
\begin{equation*}
\begin{array}{cccccccccc}
 & [ &  -g(\alpha,\beta) & : &  f(\alpha,\beta) & : & ih(\alpha) & : & -ih(\beta) & ] \\
  = & [ &  -f(\alpha+\pi,\beta+\pi) & : & g(\alpha+\pi,\beta+\pi) & : & ih(\beta+\pi) & : & -ih(\alpha+\pi) & ]  \\
  = & [ &  ih(\alpha) & : & -ih(\beta+\pi) & : & g(\beta+\pi,\alpha) & : & -f(\beta+\pi,\alpha) & ]  \\
  = & [ &  ih(\beta) & : & -ih(\alpha+\pi) & : & -f(\alpha+\pi,\beta) & : & g(\alpha+\pi,\beta) & ] .
  \end{array}
\end{equation*}
\begin{itemize}
\item Comparing entries 3, 4 of the lines 1, 2, we find:
$$
h(\alpha):h(\beta)=h(\beta+\pi):h(\alpha+\pi),
$$
so that \eqref{h sym}  follows;
\item comparing entries 1, 2 of the lines 1, 3, we find:
$$
g(\alpha,\beta):f(\alpha,\beta)=h(\alpha):h(\beta+\pi),
$$
which, together with \eqref{h sym}, implies \eqref{g thru fh};
\item comparing entries 2, 3 of the lines 1, 4, we find:
$$
f(\alpha,\beta): h(\alpha)=-h(\alpha+\pi):f(\alpha+\pi,\beta),
$$
which, together with \eqref{h sym}, implies \eqref{f sym}.
\end{itemize} 
After this, one easily checks that under the conditions derived thus far, all requirements are satisfied (express everything through $h(\alpha)$, $h(\beta)$, and $f(\alpha,\beta)$).
\end{proof}

\begin{theorem}\label{th 1}
Equation \eqref{quad eq} with the symmetry conditions \eqref{change orient}--\eqref{g thru fh} is 3D-consistent if and only if the following functional equation is satisfied:
\begin{equation}\label{eq fhh}
f(\alpha,\beta)h(\alpha)h(\beta)+f(\beta,\gamma)h(\beta)h(\gamma)+f(\gamma,\alpha)h(\gamma)h(\alpha)=f(\alpha,\beta)f(\beta,\gamma)f(\gamma,\alpha),
\end{equation}
which, due to \eqref{g thru fh}, is also equivalent to  
\begin{equation}\label{g+g+g}
g(\alpha,\beta)+g(\beta,\gamma)+g(\gamma,\alpha)=c^{-1}f(\alpha,\beta)f(\beta,\gamma)f(\gamma,\alpha).
\end{equation}
\end{theorem}

\begin{proof}
One imposes the equation on all six faces of an elementary 3D cube, see Figure \ref{fig: 3D consist}. Then one has three different answers for $x_{123}$ in terms of the initial data $x_0$, $x_1$, $x_2$, $x_3$ (actually, they will not depend on $x_0$; this is the tetrahedron property which is a consequence of the three-leg form of the equation).

\begin{figure}[htbp]
\centering
  \subfloat[3D consistency of quad equations with labeling on directed edges]{\label{fig: 3D consist}
\begin{tikzpicture}[scale=1.7]

\draw[-{Stealth[scale=1.5]}] (0,0) -- (1,0) node[below] {$\alpha$};
\draw (1,0) -- (2,0); 
\draw[-{Stealth[scale=1.5]}] (0,0) -- (0,1) node[left] {$\gamma$};
\draw (0,1) -- (0,2); 
\draw[-{Stealth[scale=1.5]}] (2,2) -- (1,2) node[above] {$\alpha + \pi$};
\draw (1,2) -- (0,2); 
\draw[-{Stealth[scale=1.5]}] (2,2) -- (2,1) node[above left] {$\gamma + \pi$};
\draw (2,1) -- (2,0); 

\draw[dashed, -{Stealth[scale=1.5]}] (0,0) -- (.3,.4) node[right] {$\beta$};
\draw[dashed] (.3,.4) -- (.6,.8); 
\draw[-{Stealth[scale=1.5]}] (2.6,.8) -- (2.3,.4) node[below right] {$\beta + \pi$};
\draw (2.3,.4) -- (2,0); 
\draw[-{Stealth[scale=1.5]}] (.6,2.8) -- (.3,2.4) node[left] {$\beta + \pi$};
\draw (.3,2.4) -- (0,2); 
\draw[-{Stealth[scale=1.5]}] (2,2) -- (2.3,2.4) node[left] {$\beta$};
\draw (2.3,2.4) -- (2.6,2.8); 

\begin{scope}[xshift=6mm,yshift=8mm]
	\draw[dashed, -{Stealth[scale=1.5]}] (2,0) -- (1,0) node[below] {$\alpha + \pi$};
	\draw[dashed] (1,0) -- (0,0); 
	\draw[dashed, -{Stealth[scale=1.5]}] (0,2) -- (0,1) node[below right] {$\gamma + \pi $};
	\draw[dashed] (0,1) -- (0,0); 
	\draw[-{Stealth[scale=1.5]}] (0,2) -- (1,2) node[below] {$\alpha$};
	\draw (1,2) -- (2,2); 
	\draw[-{Stealth[scale=1.5]}] (2,0) -- (2,1) node[right] {$\gamma$};
	\draw (2,1) -- (2,2); 
	
	\draw[fill=white] (0,0) circle(2pt) node[above left] {$x_2$};
	\draw[fill=white] (2,2) circle(2pt) node[above right] {$x_{123}$};
	\draw[fill=black] (2,0) circle(2pt) node[below right] {$x_{12}$};
	\draw[fill=black] (0,2) circle(2pt) node[above left] {$x_{23}$};
\end{scope}

\draw[fill=black] (0,0) circle(2pt) node[below left] {$x_0$};
\draw[fill=black] (2,2) circle(2pt) node[below right] {$x_{13}$};
\draw[fill=white] (2,0) circle(2pt) node[below right] {$x_1$};
\draw[fill=white] (0,2) circle(2pt) node[left] {$x_3\;$};

\end{tikzpicture}}
\qquad
\subfloat[3D corner adjacent to the vertex $x_1$]{\label{fig: 3D corner}
\begin{tikzpicture}[scale=1.7]

\draw[-{Stealth[scale=1.5]}] (0,0) -- (1,0) node[below] {$\alpha$};
\draw (1,0) -- (2,0); 
\draw[-{Stealth[scale=1.5]}] (0,0) -- (0,1) node[left] {$\gamma$};
\draw (0,1) -- (0,2); 
\draw[-{Stealth[scale=1.5]}] (2,2) -- (1,2) node[above] {$\alpha + \pi$};
\draw (1,2) -- (0,2); 
\draw[-{Stealth[scale=1.5]}] (2,2) -- (2,1) node[above left] {$\gamma + \pi$};
\draw (2,1) -- (2,0); 

\draw[dashed, -{Stealth[scale=1.5]}] (0,0) -- (.3,.4) node[right] {$\beta$};
\draw[dashed] (.3,.4) -- (.6,.8); 
\draw[-{Stealth[scale=1.5]}] (2.6,.8) -- (2.3,.4) node[below right] {$\beta + \pi$};
\draw (2.3,.4) -- (2,0); 
\draw[-{Stealth[scale=1.5]}] (2,2) -- (2.3,2.4) node[left] {$\beta$};
\draw (2.3,2.4) -- (2.6,2.8); 

\begin{scope}[xshift=6mm,yshift=8mm]
	\draw[dashed, -{Stealth[scale=1.5]}] (2,0) -- (1,0) node[below] {$\alpha + \pi$};
	\draw[dashed] (1,0) -- (0,0); 
	\draw[-{Stealth[scale=1.5]}] (2,0) -- (2,1) node[right] {$\gamma$};
	\draw (2,1) -- (2,2); 
	
	\draw[fill=white] (0,0) circle(2pt) node[above left] {$x_2$};
	\draw[fill=white] (2,2) circle(2pt) node[above right] {$x_{123}$};
	\draw[fill=black] (2,0) circle(2pt) node[below right] {$x_{12}$};
\end{scope}

\draw[fill=black] (0,0) circle(2pt) node[below left] {$x_0$};
\draw[fill=black] (2,2) circle(2pt) node[below right] {$x_{13}$};
\draw[fill=white] (2,0) circle(2pt) node[below right] {$x_1$};
\draw[fill=white] (0,2) circle(2pt) node[above left] {$x_3$};

\draw[-{>[scale=1]}] (2.24,.32) arc (53:90:.4);
\draw[-{>[scale=1]}] (2,.3) arc (90:180:.3);
\draw[dashed, -{>[scale=1]}] (1.8,0) arc (180:53:.2);

\end{tikzpicture}}
\caption{}
\label{Fig: 3D consistency and 3D corner}
\end{figure}


The first answer is obtained by using three equations adjacent to the vertex $x_1$ (and, for convenience, centered at this vertex), see Figure \ref{fig: 3D corner}.
These three equations read:
\begin{align*}
& f(\gamma+\pi,\alpha)x_3-g(\gamma+\pi,\alpha)x_1=i\big(h(\alpha)x_0-h(\gamma+\pi)x_{13}\big),\\
& f(\alpha,\beta+\pi)x_2-g(\alpha,\beta+\pi)x_1=i\big(h(\beta+\pi)x_{12}-h(\alpha)x_0\big),\\
& f(\beta+\pi,\gamma+\pi)x_{123}-g(\beta+\pi,\gamma+\pi)x_1 =i\big(h(\gamma+\pi)x_{13}-h(\beta+\pi)x_{12}\big).
\end{align*}
Adding these three equations, we see that everything on the right-hand side cancels away (as it should), and we are left with the \emph{tetrahedron equation} (i.e., equation relating $x_1$, $x_2$, $x_3$ and $x_{123}$) centered at $x_1$:
\begin{align}
f(\beta+\pi,\gamma+\pi)x_{123} & -\big(g(\gamma+\pi,\alpha)+g(\alpha,\beta+\pi)+g(\beta+\pi,\gamma+\pi)\big)x_1 \nonumber\\
&+f(\alpha,\beta+\pi)x_2+f(\gamma+\pi,\alpha)x_3=0.
\end{align}
Upon using the symmetry properties \eqref{change orient}--\eqref{g thru fh}, this is put into the form
\begin{align}
x_{123} & +\Big(\dfrac{ch(\alpha)}{f(\beta,\gamma)f(\gamma,\alpha)h(\gamma)}+\dfrac{ch(\alpha)}{f(\alpha,\beta)f(\beta,\gamma)h(\beta)}-\dfrac{c}{h(\beta)h(\gamma)}\Big)x_1 \nonumber\\
&-\dfrac{c}{f(\alpha,\beta)f(\beta,\gamma)}x_2-\dfrac{c}{f(\beta,\gamma)f(\gamma,\alpha)}x_3=0.
\end{align}
The necessary and sufficient condition for this equation to be symmetric with respect to cyclic shifts of $(1,2,3)$ and of $(\alpha,\beta,\gamma)$, is given by
$$
\dfrac{h(\alpha)}{f(\beta,\gamma)f(\gamma,\alpha)h(\gamma)}+\dfrac{h(\alpha)}{f(\alpha,\beta)f(\beta,\gamma)h(\beta)}-\dfrac{1}{h(\beta)h(\gamma)}=
-\dfrac{1}{f(\gamma,\alpha)f(\alpha,\beta)}.
$$
This is equivalent to condition \eqref{eq fhh} (which is, luckily, in its turn symmetric with respect to cyclic shifts of $(\alpha,\beta,\gamma)$).
\end{proof}

{\bf Remark.} In the 3D consistent case, the tetrahedron equation takes the symmetric form
\begin{equation}
 f(\beta,\gamma)x_1+f(\gamma,\alpha)x_2+f(\alpha,\beta)x_3-c^{-1}f(\alpha,\beta)f(\beta,\gamma)f(\gamma,\alpha)x_{123}=0.
\end{equation}
By virtue of equation \eqref{g+g+g}, this is nothing but the Laplace equation on the 3D corner centered at $x_{123}$:
\begin{equation}
 \Big(f(\beta,\gamma)x_1-g(\beta,\gamma)x_{123}\Big)+\Big(f(\gamma,\alpha)x_2-g(\gamma,\alpha)x_{123}\Big)+\Big(f(\alpha,\beta)x_3-g(\alpha,\beta)x_{123}\Big)=0.
\end{equation}

\medskip

\begin{theorem}\label{th 2}
If functions $f$, $h$ satisfy equation \eqref{eq fhh} and $f$ is skew-symmetric, as in \eqref{change orient}, then $f$ satisfies  the following functional equation:
\begin{align}\label{eq fff}
f(\alpha,\beta)f(\beta,\gamma)f(\gamma,\alpha)+f(\beta,\alpha)f(\alpha,\delta)f(\delta,\beta)&+f(\gamma,\beta)f(\beta,\delta)f(\delta,\gamma)\nonumber \\
&+f(\alpha,\gamma)f(\gamma,\delta)f(\delta,\alpha)=0.
\end{align}
Conversely, if \eqref{eq fff} is satisfied, then \eqref{eq fhh} is satisfied with $h(\alpha)=if(\alpha,\delta)$ for any $\delta$.
\end{theorem}
\begin{proof} The easiest way to see why equation \eqref{eq fff} should be satisfied is the following nice combinatorial interpretation. On triangulated surfaces with vertices carrying the fields $\alpha$, $\beta$, $\gamma$, etc., introduce a discrete 2-form $\omega$ by
$$
\omega\bigg( \tri{\alpha}{\beta}{\gamma} \bigg)=f(\alpha,\beta)f(\beta,\gamma)f(\gamma,\alpha).
$$
Then equation \eqref{g+g+g} expresses the fact that $\omega$ is exact, with $c^{-1}\omega=dg$ for the discrete 1-form $g$ taking the value $g(\alpha,\beta)$ on the (directed) edge $(\alpha,\beta)$. This is equivalent to the closeness of the 2-form $\omega$ on a tetrahedron:
$$
d\omega\bigg( \tetr{\alpha}{\beta}{\gamma}{\delta} \bigg)=0.
$$
Equation \eqref{eq fff} expresses precisely this latter condition. The last statement of the theorem is straightforward.
\end{proof}

\section{Solving functional equation \eqref{eq fff}}
\label{sect solve fff}

We solve functional equation \eqref{eq fff} in the particular case when the skew-symmetric function $f(\alpha,\beta)$ depends only on the difference of its arguments: $f(\alpha,\beta)=f(\alpha-\beta)$ (where we use the same notation $f$ for an odd function of one argument). Here function $f$ satisfies the following symmetry conditions:
\begin{equation}\label{ff}
f(-\alpha)=-f(\alpha), \quad f(\alpha)f(\alpha+\pi)=-c.
\end{equation}
The functional equations itself takes the form
\begin{align}\label{eq fff dif}
f(\alpha-\beta)f(\beta-\gamma)f(\gamma-\alpha)&+f(\beta-\alpha)f(\alpha-\delta)f(\delta-\beta)+f(\gamma-\beta)f(\beta-\delta)f(\delta-\gamma)\nonumber \\
&+f(\alpha-\gamma)f(\gamma-\delta)f(\delta-\alpha)=0.
\end{align}

\begin{theorem} \label{th 3}
Odd meromorphic functions satisfying equation \eqref{eq fff dif} and $f'(0)\neq 0$ are exactly odd meromorphic solutions of a differential equation $(f')^2=c_0f^4 +c_1f^2+c_2$.
If $c_0\neq 0$, they are given, up to a constant factor and a re-scaling of the argument by a constant complex factor, by \footnote{We use notations of Whittaker and Watson \cite{WW} for elliptic theta-functions $\vartheta_k(\alpha)=\vartheta_k(\alpha,q)=\vartheta_k(\alpha|\tau)$, so that, for instance,
$\vartheta_3(\alpha,q)=\sum_{n=-\infty}^{\infty} q^{n^2}\exp(2in\alpha)$, 
$q=\exp(\pi\tau i)$.}
$$
f(\alpha)=\frac{\vartheta_1(\alpha)}{\vartheta_2(\alpha)}.
$$ 
If $c_0=0$, the underlying elliptic curve degenerates to a rational one, and we get degenerate solutions of \eqref{eq fff dif} expressed through trigonometric or rational functions. 
\end{theorem}
\begin{proof} Differentiate equation \eqref{eq fff dif} with respect to $\alpha$ and then set $\alpha=\beta$:
\begin{align}\label{eq fff aux1}
& f'(0)f(\beta-\gamma)f(\gamma-\beta)-f'(0)f(\beta-\delta)f(\delta-\beta)\nonumber \\
&\qquad +f'(\beta-\gamma)f(\gamma-\delta)f(\delta-\beta)-f(\beta-\gamma)f(\gamma-\delta)f'(\delta-\beta)=0.
\end{align}
Re-defining $\beta-\gamma\rightsquigarrow\beta$, $\beta-\delta\rightsquigarrow\alpha$, we write this as 
\begin{equation}\label{eq fff aux2}
f'(0)\big(f^2(\alpha)-f^2(\beta)\big)=f(\alpha-\beta)\big(f(\alpha)f'(\beta)+f'(\alpha)f(\beta)\big).
\end{equation}
Changing here $\beta$ to $-\beta$, we find:
\begin{equation}\label{eq fff aux3}
f'(0)\big(f^2(\alpha)-f^2(\beta)\big)=f(\alpha+\beta)\big(f(\alpha)f'(\beta)-f'(\alpha)f(\beta)\big).
\end{equation}
As a consequence:
\begin{equation}\label{eq fff aux4}
\frac{f(\alpha+\beta)}{f(\alpha-\beta)}=\frac{f(\alpha)f'(\beta)+f'(\alpha)f(\beta)}{f(\alpha)f'(\beta)-f'(\alpha)f(\beta)}.
\end{equation}
This equation has been solved in \cite{Cooper}. For convenience of the reader, we reproduce here this solution. Expand 
$$
f(\alpha+\beta)\big(f(\alpha)f'(\beta)-f'(\alpha)f(\beta)\big)=f(\alpha-\beta)\big(f(\alpha)f'(\beta)+f'(\alpha)f(\beta)\big)
$$
about $\beta=0$, upon using $f(\beta)=a\beta+b\beta^3+\ldots$:
\begin{eqnarray*}
\lefteqn{\Big(f(\alpha)+\beta f'(\alpha)+\frac{\beta^2}{2}f''(\alpha)+\frac{\beta^3}{6}f'''(\alpha)+\ldots\Big)}\\
&&\times\Big(af(\alpha)+3b\beta^2f(\alpha)-a\beta f'(\alpha)-b\beta^3f'(\alpha)+\ldots\Big)\\
\lefteqn{=\Big(f(\alpha)-\beta f'(\alpha)+\frac{\beta^2}{2}f''(\alpha)-\frac{\beta^3}{6}f'''(\alpha)+\ldots\Big)}\\
&&\times\Big(af(\alpha)+3b\beta^2f(\alpha)+a\beta f'(\alpha)+b\beta^3f'(\alpha)+\ldots\Big).
\end{eqnarray*}
The first nontrivial consequence is obtained by comparing coefficients by $\beta^3$:
$$
\frac{a}{6}f'''(\alpha)f(\alpha)-\frac{a}{2}f''(\alpha)f'(\alpha)+2bf'(\alpha)f(\alpha)=0.
$$
Integrating once, we find:
$$
af''(\alpha)f(\alpha)-2a(f'(\alpha))^2+6bf^2(\alpha)=c,
$$
or
$$
a\left(\frac{f'(\alpha)}{f^2(\alpha)}\right)'+\frac{6b}{f(\alpha)}=\frac{c}{f^3(\alpha)}.
$$
Multiplying by $f'(\alpha)/f^2(\alpha)$ and integrating again, we find:
$$
\frac{a}{2}\left(\frac{f'(\alpha)}{f^2(\alpha)}\right)^2-\frac{3b}{f^2(\alpha)}=-\frac{c}{4f^4(\alpha)}+d.
$$
This proves that $f$ necessarily solves a differential equation $(f')^2=c_0f^4+c_1f^2+c_2$, since $a=f'(0)\neq 0$. It is well-known that, if $c_0\neq 0$, $c_2\neq 0$ and $c_1^2-4c_0c_2\neq 0$, then odd solutions of this differential equation are given, up to a constant factor and a re-scaling of the argument,  by $f(\alpha)=\vartheta_1(\alpha)/\vartheta_2(\alpha)$. 

It remains to show that this function indeed satisfies functional equation \eqref{eq fff dif}. This follows from the following fact: 
\begin{equation}\label{log der}
\frac{\vartheta'_2(\alpha-\beta)}{\vartheta_2(\alpha-\beta)}+\frac{\vartheta'_2(\beta-\gamma)}{\vartheta_2(\beta-\gamma)}+\frac{\vartheta'_2(\gamma-\alpha)}{\vartheta_2(\gamma-\alpha)}=
-\vartheta_3\vartheta_4\frac{\vartheta_1(\alpha-\beta)\vartheta_1(\beta-\gamma)\vartheta_1(\gamma-\alpha)}{\vartheta_2(\alpha-\beta)\vartheta_2(\beta-\gamma)\vartheta_2(\gamma-\alpha)}.
\end{equation}
In other words, function $f(\alpha)=\vartheta_1(\alpha)/\vartheta_2(\alpha)$ satisfies an analog of equation \eqref{g+g+g}:
\begin{equation}\label{g+g+g dif}
f(\alpha-\beta)f(\beta-\gamma)f(\gamma-\alpha)=g_0(\alpha-\beta)+g_0(\beta-\gamma)+g_0(\gamma-\alpha)
\end{equation}
with an odd function 
$$
g_0(\alpha)=-\frac{1}{\vartheta_3\vartheta_4}\cdot\frac{\vartheta'_2(\alpha)}{\vartheta_2(\alpha)}, 
$$
which immediately implies \eqref{eq fff dif}. 
\end{proof}

Real-analytic solutions turn out to correspond to two types of Riemann surfaces of genus 1 with an antiholmorphic involution -- complex tori of a rectangular and of a rhombic type.
\begin{theorem}\label{th 4}
All solutions of equation \eqref{eq fff dif} which satisfy conditions of Theorem \ref{th 3} and are real-valued for $\alpha\in\bbR$ and fulfill \eqref{ff}, are given by
\begin{equation}\label{f real}
 f(\alpha)=\frac{\vartheta_1(\alpha/2)}{\vartheta_2(\alpha/2)},
\end{equation}
where the parameter $\tau$ of the theta-functions satisfies 
\begin{itemize}
\item either $\tau=i\tau_0$ with $\tau_0\in\bbR_{>0}$, so that the complex torus generated by $\pi$, $\pi\tau$ is of the rectangular type,
\item or $\tau=\frac{1}{2}+i\tau_0$ with $\tau_0\in\bbR_{>0}$, so that the complex torus generated by $\pi$, $\pi\tau$ is of the rhombic type.
\end{itemize}
\end{theorem}
\begin{proof}
A discussion of the real-valued solutions is also contained in \cite{Cooper}. 
Solutions from Theorem \ref{th 3} which are real-valued for $\alpha\in\mathbb R$ are (up to a constant real factor)
\begin{align}
& f(\alpha)=\frac{\vartheta_1(\alpha/2)}{\vartheta_2(\alpha/2)}, \label{f1}\\
& f(\alpha)=\frac{\vartheta_1(\alpha/2)}{\vartheta_3(\alpha/2)}, \label{f2}\\
& f(\alpha)=\frac{\vartheta_1(\alpha/2)}{\vartheta_4(\alpha/2)}, \label{f3}\\ 
& f(\alpha)=\frac{\vartheta_1(\alpha/2)\vartheta_3(\alpha/2)}{\vartheta_2(\alpha/2)\vartheta_4(\alpha/2)}, \label{f4}
\end{align}
where the parameters of the theta functions satisfy $\tau\in i\mathbb R_{>0}\Leftrightarrow 0<q<1$, 
and their degenerations. Formulas \eqref{f1}--\eqref{f3} correspond to the case when the real numbers $c_0$, $c_1$, $c_2$ satisfy $c_1^2-4c_0c_2>0$; in the case when the quadratic polynomial $c_0f^2+c_1f+c_2$ has one positive and one negative root, we get solution \eqref{f1}, if this polynomial 
has two negative roots, we get solution \eqref{f2}, and if this polynomial has two positive roots, we get solution \eqref{f3}.  Solutuion \eqref{f4} corresponds to the case $c_1^2-4c_0c_2<0$, when $c_0f^2+c_1f+c_2$ has complex conjugate roots.

Finally, we observe that only two of the four real-valued solutions satisfy an additional requirement $f(\alpha)f(\alpha+\pi)=c$, namely \eqref{f1} and \eqref{f4}. Moreover,
solution \eqref{f4} can be written as \eqref{f1} with $\tau=\frac{1}{2}+i\tau_0$, $\tau_0\in\bbR_{>0}$, see the next section.
\end{proof}

\section{Properties of real-valued solutions}
\subsection{Rectangular case, $\tau=i\tau_0$, $\tau_0>0$. }\quad

As mentioned in Theorem \ref{th 2}, one can find functions $g(\alpha,\beta)=f(\alpha-\beta)h(\alpha)h(\beta)$ solving equation \eqref{g+g+g} by setting  $ih(\alpha)=f(\alpha-\delta)$ with some fixed $\delta$. The function $f(\alpha)$ takes purely imaginary values along the horizontal lines $\Im(\alpha)=\pm \pi\tau_0/2$. Thus, we set $\delta=-i\pi\tau_0/2+\lambda_0$ with $\lambda_0\in\mathbb R$, and we get a one-parameter family of admissible real-valued functions 
\begin{equation}
h(\alpha)=h(\alpha;\lambda_0)=\frac{\vartheta_4((\alpha-\lambda_0)/2)}{\vartheta_3((\alpha-\lambda_0)/2)}.
\end{equation}
We mention the estimates, valid on the real axis:
\begin{equation}\label{ineq h}
0<\frac{\vartheta_4}{\vartheta_3}\le h(\alpha)\le \frac{\vartheta_3}{\vartheta_4}.
\end{equation}
Here $\vartheta_3=\vartheta_3(0)$ and $\vartheta_4=\vartheta_4(0)$ are theta-constants. We mention also the following corollary of \eqref{log der}:
\begin{equation}\label{eq g=g+g+g 1}
g(\alpha,\beta;\lambda_0)= f(\alpha-\beta)h(\alpha;\lambda_0)h(\beta;\lambda_0)=g_0(\alpha-\beta)+g_1(\beta-\lambda_0)-g_1(\alpha-\lambda_0),
\end{equation}
where
\begin{eqnarray}
g_0(\alpha)&= &-\frac{1}{\vartheta_3\vartheta_4}\cdot \frac{\vartheta'_2(\alpha/2)}{\vartheta_2(\alpha/2)},\\
g_1(\alpha) & = &-\frac{1}{\vartheta_3\vartheta_4}\cdot \frac{\vartheta'_3(\alpha/2)}{\vartheta_3(\alpha/2)}\ = \ g_0\Big(\alpha+\frac{\pi i\tau_0}{2}\Big).
\end{eqnarray}
The additive decomposition \eqref{eq g=g+g+g 1} has the following important consequence: upon summation over the rhombi adjacent to a given black point $x_0$, as in equation \eqref{Laplace eq}, the mass coefficient is independent of $\lambda_0$:
\begin{equation}\label{sum g = sum g_0}
\sum_{k=1}^m g(\alpha_k,\alpha_{k+1};\lambda_0)=\sum_{k=1}^m g_0(\alpha_k-\alpha_{k+1}).
\end{equation}
Below, the following inequality will be of importance.
\begin{lemma} \label{lemma 2}
In the rectangular case, for all $0<\alpha<\pi$, one has:
\begin{equation}\label{ineq g_0>f}
g_0(\alpha)>f(\alpha)>0\quad \Leftrightarrow \quad -\frac{\vartheta'_2(\alpha/2)}{\vartheta_2(\alpha/2)}>\vartheta_3\vartheta_4\frac{\vartheta_1(\alpha/2)}{\vartheta_2(\alpha/2)}.
\end{equation}
\end{lemma}
\begin{proof}
This follows from the well-known Fourier series expansions  \cite{WW}:
\begin{eqnarray*}
-\frac{\vartheta'_2(\alpha)}{\vartheta_2(\alpha)} & = & \tan (\alpha)+4\sum_{n=1}^\infty (-1)^{n-1}\frac{q^{2n}}{1-q^{2n}}\sin (2n\alpha),\\
\vartheta_3\vartheta_4\frac{\vartheta_1(\alpha)}{\vartheta_2(\alpha)} & = & \tan(\alpha)+4\sum_{n=1}^\infty(-1)^n\frac{q^{2n}}{1+q^{2n}}\sin(2n\alpha).
\end{eqnarray*}
There follows:
$$
-\frac{\vartheta'_2(\alpha)}{\vartheta_2(\alpha)}-\vartheta_3\vartheta_4\frac{\vartheta_1(\alpha)}{\vartheta_2(\alpha)}  =  8\sum_{n=1}^\infty (-1)^{n-1}\frac{q^{2n}}{1-q^{4n}}\sin (2n\alpha)=-\frac{\vartheta'_3(\alpha;q^2)}{\vartheta_3(\alpha;q^2)}>0
$$
for $0<\alpha<\pi/2$.
\end{proof}
Both inequalities \eqref{ineq h} and \eqref{ineq g_0>f} will be used in Section \ref{sect Laplace} to show, in two different ways, the positivity of weights for the energy functional for the Laplacian \eqref{Laplace eq} in the case of a rectangular complex torus.

\subsection{Rhombic case, $\tau=\frac{1}{2}+i\tau_0$, $\tau_0>0$.}\quad

As mentioned in Theorem \ref{th 2}, one can find functions $g(\alpha,\beta)=f(\alpha-\beta)h(\alpha)h(\beta)$ solving equation \eqref{g+g+g} by setting  $ih(\alpha)=f(\alpha-\delta)$ with some fixed $\delta$. In the present case, the function $if(\alpha)$ is real-valued (that is, $f(\alpha)$ takes pure imaginary values) on the vertical lines $\Re(\alpha) =\pi/2$ and $\Re(\alpha)=0$. Therefore, we cannot find a real-valued $h(\alpha)$ by restricting $f(\alpha)$ to some horizontal line. Hence, we are forced to consider here the case of imaginary-valued function $h(\alpha)$, coming from a real-valued $f(\alpha-\delta)$. This can be achieved by taking $\delta=\pi\tau+\lambda_0$ with $\lambda_0\in\mathbb R$, so that $\Im(\delta)=\pi\tau_0$. Thus, a one-parameter family of admissible imaginary-valued functions $h$ is found:
$$
h(\alpha)=h(\alpha;\lambda_0)=-if(\alpha-\lambda_0-\pi\tau)=if(\alpha-\lambda_0). 
$$
We have:
\begin{eqnarray}
g(\alpha,\beta;\lambda_0)=f(\alpha-\beta)h(\alpha;\lambda_0)h(\beta;\lambda_0) & = & -f(\alpha-\beta)f(\alpha-\lambda_0)f(\beta-\lambda_0)\nonumber\\
& = & f(\alpha-\beta)f(\beta-\lambda_0)f(\lambda_0-\alpha)\nonumber\\
& = & g_0(\alpha-\beta)+g_0(\beta-\lambda_0)-g_0(\alpha-\lambda_0). \label{eq g=g+g+g 2}
\end{eqnarray}
The latter equation has the same consequence \eqref{sum g = sum g_0} as in the rectangular case.

Computation in the rhombic case are facilitated by observing that for $\tau=\frac{1}{2}+i\tau_0$, $\tau_0>0$ one has
$$
q=\exp(\pi i\tau)=iq_0, \quad q_0=\exp(-\pi\tau_0)\in (0,1).
$$
There follows:
\begin{eqnarray}
\frac{\vartheta_1(\alpha/2,q)}{\vartheta_2(\alpha/2,q)} & = & \tan (\alpha/2)\prod_{n=1}^\infty \frac{(1-q^{2n}e^{i\alpha})(1-q^{2n}e^{-i\alpha})}
                                                                                                                                                    {(1+q^{2n}e^{i\alpha})(1+q^{2n}e^{-i\alpha})} \nonumber\\
& = & \tan (\alpha/2)\prod_{n=1}^\infty \frac{(1-(-1)^nq_0^{2n}e^{i\alpha})(1-(-1)^nq_0^{2n}e^{-i\alpha})}
                                                                    {(1+(-1)^nq_0^{2n}e^{i\alpha})(1+(-1)^nq_0^{2n}e^{-i\alpha})}  \nonumber\\
& = & \tan (\alpha/2)\prod_{n=1}^\infty \frac{(1-q_0^{4n}e^{i\alpha})(1-q_0^{4n}e^{-i\alpha})(1+q_0^{4n-2}e^{i\alpha})(1+q_0^{4n-2}e^{-i\alpha})}
                                                       {(1+q_0^{4n}e^{i\alpha})(1+q_0^{4n}e^{-i\alpha})(1-q_0^{4n-2}e^{i\alpha})(1-q_0^{4n-2}e^{-i\alpha})}  \nonumber \\
& = & \frac{\vartheta_1(\alpha/2,q_0^2)\vartheta_3(\alpha/2,q_0^2)}{\vartheta_2(\alpha/2,q_0^2)\vartheta_4(\alpha/2,q_0^2)}.
\end{eqnarray}
According to \eqref{log der}, function $f(\alpha)$ satisfies equation \eqref{g+g+g dif} with
$$
g_0(\alpha)=-\frac{1}{\vartheta_3(0,q)\vartheta_4(0,q)}\cdot\frac{\vartheta'_2(\alpha/2,q)}{\vartheta_2(\alpha/2,q)}.
$$
\begin{lemma}\label{lemma 3}
In the rhombic case, for all $0<\alpha<\pi$, one has:
\begin{equation}
f(\alpha)>g_0(\alpha)>0 \quad \Leftrightarrow\quad \frac{\vartheta_1(\alpha/2)}{\vartheta_2(\alpha/2)}>-\frac{1}{\vartheta_3\vartheta_4}\cdot\frac{\vartheta'_2(\alpha/2)}{\vartheta_2(\alpha/2)}.
\end{equation}
\end{lemma}
\begin{proof}
\begin{eqnarray}
\vartheta_3(0,q)\vartheta_4(0,q) & = & \prod_{n=1}^\infty(1-q^{2n})(1+q^{2n-1})^2\prod_{n=1}^\infty(1-q^{2n})(1-q^{2n-1})^2 \nonumber\\
& = & \prod_{n=1}^\infty(1-q^{2n})^2(1-q^{4n-2})^2 \nonumber\\
& = & \prod_{n=1}^\infty(1-(-1)^nq_0^{2n})^2(1+q_0^{4n-2})^2\nonumber\\
& = & \prod_{n=1}^\infty(1-q_0^{4n})^2(1+q_0^{4n-2})^2(1+q_0^{4n-2})^2\nonumber\\
& = & \prod_{n=1}^\infty(1-q_0^{4n})^2(1+q_0^{4n-2})^4\nonumber\\
& = & \vartheta_3^2(0,q_0^2).
\end{eqnarray}
We have the following well-known Fourier expansions \cite{WW}:
\begin{eqnarray}
-\frac{\vartheta'_2(\alpha,q)}{\vartheta_2(\alpha,q)} & = & \tan \alpha-4\sum_{n=1}^\infty\frac{(-1)^nq^{2n}}{1-q^{2n}}\sin(2n\alpha) \nonumber\\
& = &  \tan \alpha-4\sum_{n=1}^\infty\frac{q_0^{2n}}{1-(-1)^nq_0^{2n}}\sin(2n\alpha),
\end{eqnarray}
and
\begin{eqnarray}
\vartheta_3^2(0,q_0^2)\frac{\vartheta_1(\alpha,q_0^2)\vartheta_3(\alpha,q_0^2)}{\vartheta_2(\alpha,q_0^2)\vartheta_4(\alpha,q_0^2)} 
& = & \tan \alpha+4\sum_{n=1}^\infty\frac{q_0^{2n}}{1+(-1)^nq_0^{2n}}\sin(2n\alpha).
\end{eqnarray}
Thus, we have:
\begin{eqnarray*}
\vartheta_3^2(0,q_0^2)\big(f(\alpha)-g_0(\alpha)\big) & = & 4\sum_{n=1}^\infty\frac{q_0^{2n}}{1+(-1)^nq_0^{2n}}\sin(n\alpha)+4\sum_{n=1}^\infty\frac{q_0^{2n}}{1-(-1)^nq_0^{2n}}\sin(n\alpha)\\
 & = & 8\sum_{n=1}^\infty \frac{q_0^{2n}}{1-q_0^{4n}}\sin(n\alpha) =2\frac{\vartheta'_4(\alpha/2,q_0^2)}{\vartheta_4(\alpha/2,q_0^2)},
\end{eqnarray*}
and this is positive for $0<\alpha<\pi$.
\end{proof}

\section{Discrete exponential function}
\label{sect exp}

It is well known \cite{BS_quad, BS_DDG_book} that the 3D consistency of a quad-equation implies the existence of \emph{B\"acklund transformations}. More precisely, consider a multi-dimensionally consistent quad-equation on a quad-graph $G$. Then for any solution $x:V(G)\to\bbC$, there is a two-parameter family of solutions $x^+:V(G)\to\bbC$ called \emph{B\"acklund transforms} of $x$ and constructed as follows. One formally extends the planar quad-graph $G$ into the third dimension, by considering the second copy $G^+$ of $G$ and adding edges connecting each vertex $v\in V(G)$ with its copy $v^+\in V(G^+)$. On this way we obtain a two-layer ``3D quad-graph'' ${\bf G}$. Elementary building blocks of $\bf G$ are combinatorial cubes $(v,v_1,v_{12},v_2,v^+,v_1^+,v_{12}^+,v_2^+)$. One requires that the quad-equation be satisfied on all faces of $\bf G$, which is possible due to the 3D consistency. Thus, a B\"acklund transformation is obtained effectively by imposing the quad-equation on all ``vertical'' quads $(v,v_i,v_i^+,v^+)$ for all edges $(v,v_i)$ of $G$. It is determined by its value at one vertex of $G$ and by the parameter $\lambda$ assigned to vertical edges. 

In the present context of labels assigned to directed edges of a bipartite quad-graph $G$, the combinatorics of an elementary building block of the B\"acklund transformation is shown on Figure \ref{Fig Backlund cube}.

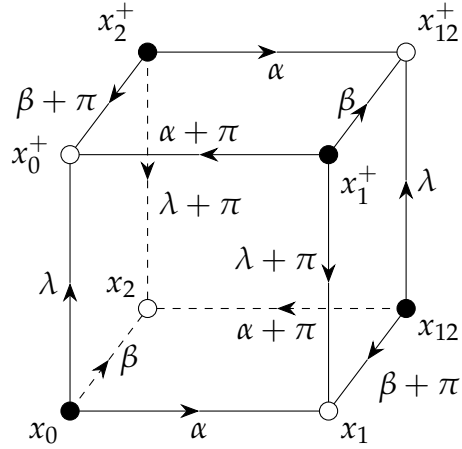
\begin{figure}[htbp]
\centering
\begin{tikzpicture}[scale=1.7]

\draw[-{Stealth[scale=1.5]}] (0,0) -- (1,0) node[below] {$\alpha$};
\draw (1,0) -- (2,0); 
\draw[-{Stealth[scale=1.5]}] (0,0) -- (0,1) node[left] {$\lambda$};
\draw (0,1) -- (0,2); 
\draw[-{Stealth[scale=1.5]}] (2,2) -- (1,2) node[above] {$\alpha + \pi$};
\draw (1,2) -- (0,2); 
\draw[-{Stealth[scale=1.5]}] (2,2) -- (2,1) node[above left] {$\lambda + \pi$};
\draw (2,1) -- (2,0); 

\draw[dashed, -{Stealth[scale=1.5]}] (0,0) -- (.3,.4) node[right] {$\beta$};
\draw[dashed] (.3,.4) -- (.6,.8); 
\draw[-{Stealth[scale=1.5]}] (2.6,.8) -- (2.3,.4) node[below right] {$\beta + \pi$};
\draw (2.3,.4) -- (2,0); 
\draw[-{Stealth[scale=1.5]}] (.6,2.8) -- (.3,2.4) node[left] {$\beta + \pi$};
\draw (.3,2.4) -- (0,2); 
\draw[-{Stealth[scale=1.5]}] (2,2) -- (2.3,2.4) node[left] {$\beta$};
\draw (2.3,2.4) -- (2.6,2.8); 

\begin{scope}[xshift=6mm,yshift=8mm]
	\draw[dashed, -{Stealth[scale=1.5]}] (2,0) -- (1,0) node[below] {$\alpha + \pi$};
	\draw[dashed] (1,0) -- (0,0); 
	\draw[dashed, -{Stealth[scale=1.5]}] (0,2) -- (0,1) node[below right] {$\lambda + \pi $};
	\draw[dashed] (0,1) -- (0,0); 
	\draw[-{Stealth[scale=1.5]}] (0,2) -- (1,2) node[below] {$\alpha$};
	\draw (1,2) -- (2,2); 
	\draw[-{Stealth[scale=1.5]}] (2,0) -- (2,1) node[right] {$\lambda$};
	\draw (2,1) -- (2,2); 
	
	\draw[fill=white] (0,0) circle(2pt) node[above left] {$x_2$};
	\draw[fill=white] (2,2) circle(2pt) node[above right] {$x_{12}^+$};
	\draw[fill=black] (2,0) circle(2pt) node[below right] {$x_{12}$};
	\draw[fill=black] (0,2) circle(2pt) node[above left] {$x_{2}^+$};
\end{scope}

\draw[fill=black] (0,0) circle(2pt) node[below left] {$x_0$};
\draw[fill=black] (2,2) circle(2pt) node[below right] {$x_{1}^+$};
\draw[fill=white] (2,0) circle(2pt) node[below right] {$x_1$};
\draw[fill=white] (0,2) circle(2pt) node[left] {$x_0^+\;$};

\end{tikzpicture}
\caption{An elementary cube of a B\"acklund transformation with the B\"acklund parameter $\lambda$}
\label{Fig Backlund cube}
\end{figure}


\begin{definition} \label{def exp}
A \emph{discrete exponential function} for a given 3D consistent linear quad-equation is the B\"acklund transform of a zero solution.
\end{definition}
The discrete exponential function can be recursively obtained by extending it over edges according to the basic quad-equation. 

For edges $(v_0,v_1)$, $(v_0,v_2)$ pointing from a black vertex to a white one, we have to apply the quad-equation like
$$
f(\alpha,\lambda)x_1^+-g(\alpha,\lambda)x_0=i\Big(h(\lambda)x_0^+-h(\alpha)x_1\Big),
$$
with the background solution $x_0=x_1=0$. Thus, we have:
\begin{equation}
x_1^+=i\frac{h(\lambda)}{f(\alpha,\lambda)} x_0^+, \quad x_2^+=i\frac{h(\lambda)}{f(\beta,\lambda)} x_0^+.
\end{equation}
For edges $(v_1,v_{12})$, $(v_2,v_{12})$ pointing from a white vertex to a black one, the suitable quad-equation is as follows (look up the corresponding quads on Figure \ref{Fig Backlund cube}):
$$
f(\beta+\pi,\lambda+\pi)x_{12}^+-g(\beta+\pi,\lambda+\pi)x_1=i\Big(h(\lambda+\pi)x_1^+-h(\beta+\pi)x_{12}\Big),
$$
with the background solution $x_1=x_{12}=0$. Taking into account that $f(\beta+\pi,\lambda+\pi)=f(\beta,\lambda)$, we find:
\begin{align}
x_{12}^+ & =i\ \frac{h(\lambda+\pi)}{f(\beta,\lambda)} x_1^+=i\ \frac{h(\lambda+\pi)}{f(\alpha,\lambda)} x_2^+\nonumber\\
               & =-\frac{h(\lambda)h(\lambda+\pi)}{f(\alpha,\lambda)f(\beta,\lambda)} x_0^+=-\frac{c}{f(\alpha,\lambda)f(\beta,\lambda)} .
\end{align}
We come to the following result.

\begin{theorem}
Let 
$$
f(\alpha,\beta)=f(\alpha-\beta)=\frac{\vartheta_1((\alpha-\beta)/2)}{\vartheta_2((\alpha-\beta)/2)}, \quad h(\lambda)=\frac{\vartheta_4((\lambda-\lambda_0)/2)}{\vartheta_3((\lambda-\lambda_0)/2)},
$$
so that $c=1$. The discrete exponential function $\textbf{e}: V(G)\to\bbC$ normalized at some black vertex $v_0\in B$ by $\textbf{e}(v_0)=1$, is given by the following formula. For any $v\in V(G)$, consider a path $(v_0,v_1), (v_1,v_2), \ldots , (v_{m-1},v_m)$ of directed rhombus edges connecting $v_0$ with $v=v_m$, with $v_k-v_{k-1}=\exp(i\alpha_k)$. Then
$$
\textbf{e}(v)=i^m\left\{\begin{array}{cc} 1, & v\in B \\ h(\lambda), & v\in W \end{array}\right\} \cdot \prod_{k=1}^m \frac{1}{f(\alpha_k-\lambda)}.
$$ 
\end{theorem}
We see that $\emph{\textbf{e}}(v)$ is a discrete holomorphic function, real at black vertices and purely imaginary at white vertices. Its restriction to black vertices is a discrete harmonic function and was introduced in \cite{BTR}. Like in the case of discrete complex analysis, the discrete exponential functions can be shown to form a basis of discrete holomorphic functions and can be used to give an integral representation of the Green function, cf. \cite{BMS}, \cite{Kenyon}.


\section{Laplace equation and Dirichlet energy}
\label{sect Laplace}

Ansatz \eqref{quad eq} ensures that, upon summation over quads adjacent to a given (black) vertex $x_0$ (with the labels $\alpha_k$, $\alpha_{k+1}$, say) one ends up with a Laplace type equation
\begin{equation}\label{Laplace eq again}
\sum_{k=1}^m f(\alpha_k,\alpha_{k+1})x_{k,k+1}-\Big(\sum_{k=1}^m g(\alpha_k,\alpha_{k+1})\Big)x_0=0.
\end{equation}
We recall that due to \eqref{sum g = sum g_0} this can be equivalently put as 
\begin{equation}\label{Laplace eq g0}
\sum_{k=1}^m f(\alpha_k,\alpha_{k+1})x_{k,k+1}-\Big(\sum_{k=1}^m g_0(\alpha_k-\alpha_{k+1})\Big)x_0=0.
\end{equation}
These two forms of the Laplace type equation suggest two different expressions for the  Dirichlet energy functional, namely
\begin{equation}\label{Dirichlet gg}
S=\sum_{{\rm quads}\;(x_0,x_1,x_{12},x_2)} \left(\frac{1}{2}g(\alpha,\beta)x_0^2+\frac{1}{2}g(\alpha+\pi,\beta+\pi)x_{12}^2-f(\alpha,\beta)x_0x_{12}\right)
\end{equation}
and 
\begin{equation}\label{Dirichlet g0}
S=\sum_{{\rm quads}\;(x_0,x_1,x_{12},x_2)} \left(\frac{1}{2}g_0(\alpha-\beta)(x_0^2+x_{12}^2)-f(\alpha,\beta)x_0x_{12}\right),
\end{equation}
respectively. In the latter two formulas, we assume that all quads $(x_0,x_1,x_{12},x_2)$ are positively oriented. The requirements for these functionals is the symmetry of the contribution of the quad $(x_0,x_1,x_{12},x_2)$ with respect to the interchanging the roles of $x_0$ and $x_{12}$. We recall that $f(\alpha,\beta)=f(\alpha+\pi,\beta+\pi)=f(\alpha-\beta)$.
\begin{theorem}\label{th 5}
In the rectangular case, each term in each of the Dirichlet functionals \eqref{Dirichlet gg} and \eqref{Dirichlet g0} is a positive definite quadratic form.
\end{theorem}
\begin{proof}
As for functional \eqref{Dirichlet gg}, we observe that due to \eqref{g thru fh}, \eqref{h sym}, it can be put as 
\begin{eqnarray} \label{eq 2-form complete square}
S &= & \sum_{{\rm quads}\;(x_0,x_1,x_{12},x_2)} \left(\frac{1}{2c}f(\alpha,\beta)h(\alpha)h(\beta)x_0^2+\frac{c}{2}\frac{f(\alpha,\beta)}{h(\alpha)h(\beta)}x_{12}^2-f(\alpha,\beta)x_0x_{12}\right) \nonumber \\
 & = & \sum_{{\rm quads}\;(x_0,x_1,x_{12},x_2)} \frac{1}{2}f(\alpha,\beta)\left(\frac{h(\alpha)h(\beta)}{c}x_0^2+\frac{c}{h(\alpha)h(\beta)}x_{12}^2-2x_0x_{12}\right).
\end{eqnarray} 
Each term here is a complete square, with the sign of $f(\alpha-\beta)h(\alpha)h(\beta)>0$, see equation \eqref{ineq h}.

For functional \eqref{Dirichlet g0}, positivity follows from $g_0^2(\alpha-\beta)>f^2(\alpha-\beta)$, see Lemma \ref{lemma 2}.
\end{proof}

Laplacian \eqref{Laplace eq g0}, Dirichlet energy \eqref{Dirichlet g0}, as well as the positivity statement for this functional were found in \cite{BTR}. In the rhombic case, neither of Dirichlet energies \eqref{Dirichlet gg}, \eqref{Dirichlet g0} consists of positive definite quad contributions: the signs of $h(\alpha)h(\beta)$ can be arbitrary, and the contributions in \eqref{Dirichlet g0} are indefinite forms, due to Lemma \ref{lemma 3}.

\section{Discrete pluri-Lagrangian problem}
\label{sect pluri}

It turns out that the Dirichlet energies of Section \ref{sect Laplace} are very special cases of certain integrable variational discrete systems. Recall that a suitable notion of integrability of variational systems is that of a pluri-Lagrangian structure \cite{BPS2, BS_pluriharmonic}. For discrete systems, this structure can be seen as an analogue of multi-dimensional consistency in the variational context. We give here a very brief introduction to pluri-Lagrangian systems, following \cite{BS_pluriharmonic}. In considering multi-dimensional pluri-Lagrangian systems, we use the following notations. Independent discrete
variables on the lattice $\bbZ^m$ are denoted $n=(n_1,n_2,\dots,n_m)$. The lattice shifts are denoted by
$T_i:n\to n+ e_i$, where $e_i$ is the $i$-th coordinate vector.
\begin{itemize}
\item
Single superscript $i$ means that an object is associated with the edge $(n,n+e_i)$
and double superscript $ij$ is used for objects associated with the plaquette
$\sigma^{ij}=(n,n+e_i,n+e_i+e_j,n+e_j)$.
\item
Subscripts  denote shifts on the lattice. For instance, the fields at the corners of the plaquette
$\sigma^{ij}$ are denoted by $x, x_i, x_{ij}, x_j$.
\end{itemize}

In the most general form, consistent variational equations can be described as
follows.

\begin{definition}\label{def:pluriLagr problem} {\bf (2D pluri-Lagrangian problem)}

\begin{itemize}
\item Let $\cL$ be a discrete 2-form, i.e., a real-valued function of oriented elementary squares
\[
\sigma^{ij}=\left(n,n+e_{i},n+e_{i}+e_{j},n+e_{j}\right)
\]
of $\bbZ^m$, such that $\cL\left(\sigma^{ij}\right)=-\cL\left(\sigma^{ji}\right)$. We will assume that $\cL$ depends on some field assigned to the points of $\bbZ^m$, that is, on some $x:\bbZ^m\to \bbC$. More precisely, $\cL(\sigma^{ij})$ depends on the values of $x$ at the four vertices of $\sigma^{ij}$:
\begin{equation}\label{L-v}
 \cL(\sigma^{ij})=L(x,x_i,x_j,x_{ij}).
\end{equation}

\item To an arbitrary oriented quad-surface $\Sigma$ in $\bbZ^m$, there corresponds the {\em action functional}, which assigns to $x |_{V(\Sigma)}$, i.e., to the fields at the vertices of the surface $\Sigma$, the number
\begin{equation}\label{action on surface}
S_\Sigma=\sum_{\sigma^{ij}\in\Sigma}\cL(\sigma^{ij}).
\end{equation}
\item We say that the field $x:V(\Sigma)\to \bbC$ is a critical point of $S_\Sigma$, if at any interior point $n\in V(\Sigma)$, we have
\begin{equation}\label{eq: dEL gen}
    \frac{\partial S_\Sigma}{\partial x(n)}=0.
\end{equation}
Equations (\ref{eq: dEL gen}) are called {\em discrete Euler-Lagrange equations} for the action $S_\Sigma$ (or just for the surface $\Sigma$, if it is clear which 2-form $\cL$ we are speaking about).
\item We say that the field $x:\bbZ^m\to\bbC$ solves the {\em pluri-Lagrangian problem} for the Lagrangian 2-form $\cL$ if, {\em for any quad-surface $\Sigma$ in $\bbZ^m$}, the restriction $x |_{V(\Sigma)}$ is a critical point of the corresponding action $S_\Sigma$.
\end{itemize}
\end{definition}

\begin{definition}\label{def:pluriLagr system} {\bf (System of corner equations)}
A 3D-corner is a quad-surface consisting of three elementary squares adjacent to a vertex of valence 3.
The {\em system of corner equations} for a given discrete 2-form $\cL$ consists of discrete Euler-Lagrange equations at the 3-valent interior points of all possible 3D-corners in $\bbZ^m$. If the action for the surface of an oriented elementary cube $\sigma^{ijk}$ of the coordinate directions $i,j,k$ (which can be identified with the discrete exterior derivative $d\cL$ evaluated at $\sigma^{ijk}$) is denoted by ($T_k$ is the shift in the $k$-th coordinate direction)
\begin{equation}\label{eq: Sijk}
S^{ijk}=d\cL(\sigma^{ijk})=(T_k-I)\cL(\sigma^{ij})+(T_i-I)\cL(\sigma^{jk})+(T_j-I)\cL(\sigma^{ki}),
\end{equation}
then the system of corner equations consists of the eight equations
\begin{equation}\label{eq: corner eqs}
\begin{array}{llll}
\dfrac{\partial S^{ijk}}{\partial x}=0, & \dfrac{\partial S^{ijk}}{\partial x_i}=0, & \dfrac{\partial S^{ijk}}{\partial x_j}=0, & \dfrac{\partial S^{ijk}}{\partial x_k}=0, \\
\\
\dfrac{\partial S^{ijk}}{\partial x_{ij}}=0, & \dfrac{\partial S^{ijk}}{\partial x_{jk}}=0, & \dfrac{\partial S^{ijk}}{\partial x_{ik}}=0, & \dfrac{\partial S^{ijk}}{\partial x_{ijk}}=0,
\end{array}
\end{equation}
for each triple $i,j,k$. Symbolically, this can be put as $\delta(d\cL)=0$, where $\delta$ stands for the ``vertical'' differential (differential with respect to the dependent field variables $u$).
\end{definition}

\textbf{Remark.}
Corner equations (\ref{eq: corner eqs}) may be not uniquely solvable with respect to each of the fields, i.e. generically define correspondences. However, in the case of linear pluri-Lagrangian systems considered in the present paper, all corner equations are linear and this problem does not appear.
 \medskip

As demonstrated in \cite{BPS2}, the vertex star of any interior vertex of an oriented quad-surface in $\bbZ^m$ can be represented as a sum of (oriented) 3D-corners in $\bbZ^{m+1}$. Thus, the system of corner equations encompasses all possible discrete Euler-Lagrange equations for all possible quad-surfaces. In other words, solutions of a pluri-Lagrangian problem as introduced in Definition \ref{def:pluriLagr problem} are precisely solutions of the corresponding system of corner equations.
\medskip

\begin{definition}\label{def:L-consistency'} {\bf (Consistency of a pluri-Lagrangian problem)}
A pluri-Lagrangian problem with the 2-form $\cL$ is called {\em consistent} if, for any quad-surface $\Sigma$ flippable to the whole of $\bbZ^m$ (i.e., such that its images under sequences of elementary cube flips reach any point of $\bbZ^m$), any generic solution of Euler-Lagrange equations for $\Sigma$ can be extended to a solution of the system of corner equations on the whole $\bbZ^m$.
\end{definition}

Like in the case of quad-equations \cite{BS_quad, BS_DDG_book}, consistency is actually a local issue to be addressed for one elementary cube. The system of corner equations (\ref{eq: corner eqs}) for one elementary cube is heavily overdetermined. It consists of eight equations, each one connecting seven fields out of eight. Any six fields can serve as independent data, then one can use two of the corner equations to compute the remaining two fields, and the remaining six corner equations have to be satisfied identically. This justifies the following definition.

\begin{definition}\label{def: corner eqs consist} {\bf (Consistency of corner equations)} System (\ref{eq: corner eqs}) is called {\em consistent}, if it has the minimal possible rank 2, i.e., if exactly two of these equations are independent.
\end{definition}

The main feature of this definition is that the ``almost closedness'' of the 2-form $\cL$ on solutions of the system of corner equations is, so to say, built-in from the outset. 

\begin{theorem}\label{Th: almost closed}
For any triple of the coordinate directions $i,j,k$, the action $S^{ijk}$ over an elementary cube of these coordinate directions is a constant independent on the solution of the system of corner equations (\ref{eq: corner eqs}):
\[
S^{ijk}(x,\ldots,x_{ijk})=c^{ijk}={\rm const}  \pmod{\partial S^{ijk}/\partial x=0,\ \ldots,\
\partial S^{ijk}/\partial x_{ijk}=0}.
\]
\end{theorem}
\begin{proof} On the connected six-dimensional manifold of solutions, the gradient of $S^{ijk}$ considered as a function of eight variables, vanishes by virtue of (\ref{eq: corner eqs}).
\end{proof}

\begin{definition} \label{def: pluriharmonic} {\bf (Discrete pluriharmonic functions)}

If all $\cL(\sigma^{ij})=L(x,x_i,x_j,x_{ij})$ are quadratic forms of their arguments, then the action functional $S_\Sigma$ is called the {\em Dirichlet energy} corresponding to the quad-surface $\Sigma$. We call a solution $x:\bbZ^m\to\bbC$ of the pluri-Lagrangian problem in this case {\em a discrete pluriharmonic function}.
\end{definition}

Of course, this definition is not an immediate discretization of the classical notion of pluriharmonic functions on $\bbC^m$, and might be therefore misleading. However, we believe that these two notions are close in spirit and we hope that the further developments will establish a closer relation between the classical and the discrete pluriharmonicity, including approximation theorems for harmonic functions through discrete harmonic functions, cf. \cite{Chelkak_Smirnov}.

In the context of Definition \ref{def: pluriharmonic}, functional independence is replaced by linear independence, and the rank is understood in the sense of linear algebra.

\begin{theorem}\label{Th: closed}
The 2-form $\cL$ is closed on pluriharmonic functions. Let $x:\bbZ^m\to\bbR$ be a discrete pluriharmonic function, and let $\Sigma$, $\widetilde\Sigma$ be two quad-surfaces with the same boundary, then the Dirichlet energies of harmonic functions $x |_\Sigma$ and $x |_{\widetilde\Sigma}$ coincide.
\end{theorem}
\begin{proof} The constants from Theorem \ref{Th: almost closed} can be evaluated on the trivial pluriharmonic function $x\equiv 0$, which gives $c^{ijk}=0$.
\end{proof}


\section{Discrete two-forms depending on black vertices only}
\label{sect pluri bw}

In the present paper, we are mainly interested in a particular class of discrete two-forms generalizing \eqref{Dirichlet g0} and \eqref{Dirichlet gg}. The generalization of  \eqref{Dirichlet g0} leads to a novel integrable 3D system, while the generalization of \eqref{Dirichlet gg} turns out to be a non-essential extension of the well known integrable 3D system, the star-triangle map. In the present section, we work out the novel system generalizing \eqref{Dirichlet g0}, while the extension of the star-triangle map generalizing \eqref{Dirichlet gg} is briefly discussed in Appendix \ref{Appendix A}. We call the points of the even sublattice $$\bbZ^m_{\rm even}=\{n\in\bbZ^m: |n|=n_1+\ldots+n_m=0\pmod 2\}$$ black, and the points of the odd sublattice $$\bbZ^m_{\rm odd}=\{n\in\bbZ^m: |n|=n_1+\ldots+n_m=1\pmod 2\}$$ white. There are two kinds of elementary squares $\sigma^{ij}=(n,n+e_i,n+e_i+e_j,n+e_j)$, depending on the color of the point $n$:
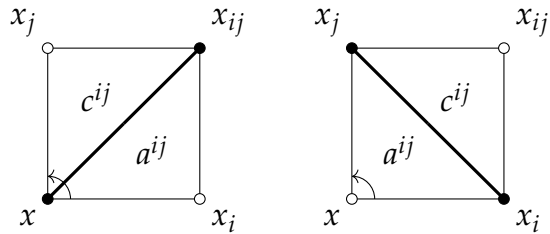
\begin{figure}[h]
\centering
\begin{tikzpicture}[scale=1]

\draw (0,0) -- (2,0); 
\draw (0,0) -- (0,2); 
\draw (2,2) -- (0,2); 
\draw (2,2) -- (2,0); 
\draw[very thick] (0,0)-- (1,1)
node[below right] {$a^{ij}$} node[above left] {$c^{ij}$};
\draw[very thick] (1,1)-- (2,2);

\draw[fill=black] (0,0) circle(2pt) node[below left] {$x$};
\draw[fill=black] (2,2) circle(2pt) node[above right] {$x_{ij}$};
\draw[fill=white] (2,0) circle(2pt) node[below right] {$x_i$};
\draw[fill=white] (0,2) circle(2pt) node[above left] {$x_j$};

\draw[-{>[scale=1]}] (.3,0) arc (0:90:.3);

\begin{scope}[xshift=40mm]
	\draw (0,0) -- (2,0); 
	\draw (0,0) -- (0,2); 
	\draw (2,2) -- (0,2); 
	\draw (2,2) -- (2,0);
	\draw[very thick] (0,2)-- (1,1)
	node[below left] {$a^{ij}$} node[above right] {$c^{ij}$};
	\draw[very thick] (2,0)-- (1,1);

	\draw[fill=white] (0,0) circle(2pt) node[below left] {$x$};
	\draw[fill=white] (2,2) circle(2pt) node[above right] {$x_{ij}$};
	\draw[fill=black] (2,0) circle(2pt) node[below right] {$x_i$};
	\draw[fill=black] (0,2) circle(2pt) node[above left] {$x_j$};
	
	\draw[-{>[scale=1]}] (.3,0) arc (0:90:.3);
\end{scope}
\end{tikzpicture}
\caption[quad-equation]{Two types of 2-cells of a bi-colored square lattice}
\label{fig: two-cells}
\end{figure}

If the points $n$ and $n+e_i+e_j$ are black, we set:
\begin{equation} \label{eq: 2-form bw 1}
\cL(\sigma^{ij})=\frac{1}{2}c^{ij}(x^2+x_{ij}^2)-a^{ij}xx_{ij}.
\end{equation}
Similarly, if the points $n+e_i$ and $n+e_j$ are black, we set:
\begin{equation} \label{eq: 2-form bw 2}
\cL(\sigma^{ij})=\frac{1}{2}c^{ij}(x_i^2+x_j^2)-a^{ij}x_ix_j.
\end{equation}
In both cases, $a^{ij}$, $c^{ij}$ are numbers assigned to the plaquettes $\sigma^{ij}$ (or to their black diagonals) and satisfy $a^{ij}=-a^{ji}$, $c^{ij}=-c^{ji}$. 

\textbf{Remark.} This ansatz is a generalization of the 2-form for the discrete complex analysis, which corresponds to the particular case
$$
a^{ij}=c^{ij}.
$$
This situation was considered in great detail in \cite{BS_pluriharmonic}. Note that under this condition, the two-form \eqref{eq: 2-form bw 1}, \eqref{eq: 2-form bw 2} becomes
\begin{equation} \label{eq: 2-form bw diag}
\cL(\sigma^{ij})=\frac{1}{2}a^{ij}(x_{ij}-x)^2, \quad \textrm{resp.} \quad \cL(\sigma^{ij})=\frac{1}{2}a^{ij}(x_j-x_i)^2,
\end{equation}
thus depends only on the differences of fields along (black) diagonals of $\sigma^{ij}$; paper \cite{BS_pluriharmonic} was devoted to the classification of pluri-Lagrangian systems with the latter property.
\medskip

For discrete 2-forms \eqref{eq: 2-form bw 1}, \eqref{eq: 2-form bw 2}, the exterior derivative $S^{ijk}=d\cL(\sigma^{ijk})$ given in \eqref{eq: Sijk} depends only on four black vertices of the elementary cube $\sigma^{ijk}$, so that the system of corner equations \eqref{eq: corner eqs} consists of four equations, each one connecting four fields at the black points. For definiteness, we write down this system for the case when the point $n$ is black. Taking into account the orientations of the sides of the elementary cube, we have:
\begin{align}
S^{ijk}=& \frac{1}{2}c^{ij}_k(x_{ik}^2+x_{jk}^2)-a^{ij}_k x_{ik}x_{jk}+\frac{1}{2}c^{jk}_i(x_{ij}^2+x_{ki}^2)-a^{jk}_i x_{ij}x_{ki}+\frac{1}{2}c^{ki}_j(x_{ij}^2+x_{jk}^2)-a^{ki}_j x_{ij}x_{jk} \nonumber\\
 &  -\frac{1}{2}c^{ij}(x^2+x_{ij}^2)+a^{ij}xx_{ij}-\frac{1}{2}c^{jk}(x^2+x_{jk}^2)+a^{jk}xx_{jk}-\frac{1}{2}c^{ki}(x^2+x_{ki}^2)+a^{ki}xx_{ki}.
\end{align}
Therefore, the system of corner equations consists of
\begin{eqnarray}
\frac{\partial S^{ijk}}{\partial x} &=& -(c^{ij}+c^{jk}+c^{ki})x+a^{ij}x_{ij}+a^{jk}x_{jk}+a^{ki}x_{ki}=0,\\
\frac{\partial S^{ijk}}{\partial x_{ij}} &=& a^{ij}x+(c^{jk}_i+c^{ki}_j-c^{ij})x_{ij}-a^{ki}_jx_{jk}-a^{jk}_ix_{ki}=0, \quad
\end{eqnarray}
as well as of two further equations obtained from the last one by cyclic shifts of the indices $(i,j,k)$. Each of these four equations is a linear equation relating the four fields at four black points of the cube. In this case, Definition \ref{def: corner eqs consist} has to be modified as follows: the system of corner equations must have rank 1, i.e., all four equations must be proportional.

\begin{theorem}
Lagrangian 2-forms \eqref{eq: 2-form bw 1}, \eqref{eq: 2-form bw 2} are consistent, i.e., form a pluri-Lagrangian system, if and only if their coefficients satisfy
\begin{eqnarray}
a^{ij}_k & = & \frac{a^{jk}a^{ki}}{c^{ij}+c^{jk}+c^{ki}}, \label{eq: a tilde} \\
c^{ij}_k & = & \frac{1}{2}\left(c^{jk}+c^{ki}-c^{ij}-\frac{(a^{jk})^2+(a^{ki})^2-(a^{ij})^2}{c^{ij}+c^{jk}+c^{ki}}\right). \label{eq: c tilde}
\end{eqnarray}
The rational map $F:(a^{ij},c^{ij},a^{jk},c^{jk},a^{ki},c^{ki})\mapsto(a^{ij}_k,c^{ij}_k,a^{jk}_i,c^{jk}_i,a^{ki}_j,c^{ki}_j)$ will be called \emph{two-field star-triangle map}, see Figure \ref{fig: star-triangle}.
\end{theorem}
\begin{proof}
A straightforward analysis of the rank 1 conditions.
\end{proof}

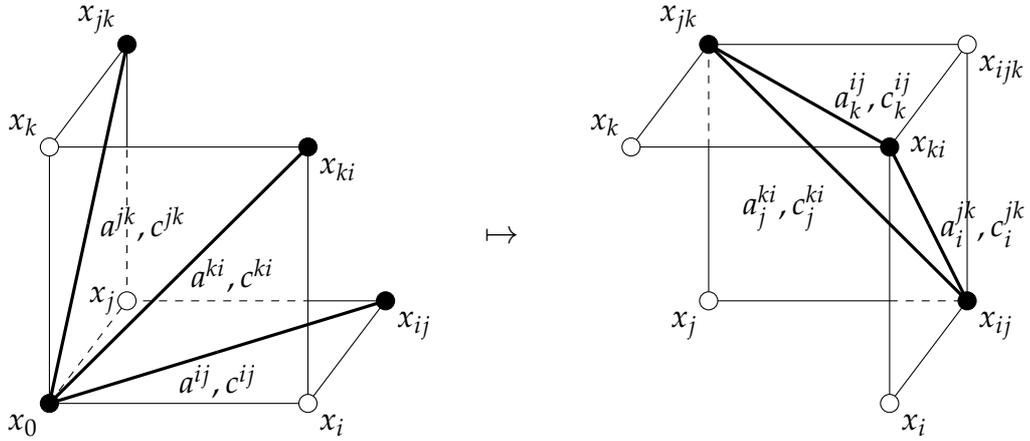
\begin{figure}[htbp]
\centering
\begin{tikzpicture}[scale=1.7]

\draw (0,0) -- (2,0); 
\draw (0,0) -- (0,2); 
\draw (2,2) -- (0,2); 
\draw (2,2) -- (2,0); 

\draw[dashed] (0,0) -- (.6,.8); 
\draw (2.6,.8) -- (2,0); 
\draw (0,2) -- (0.6,2.8); 

\draw[dashed] (0.6,0.8) -- (2,0.8); \draw (2,0.8) -- (2.6,0.8);
\draw[dashed] (0.6,0.8) -- (0.6,2.0);  \draw (0.6,2.0) -- (0.6, 2.8);

\draw[very thick] (0,0) -- (0.3,1.4) node[right]  {$a^{jk},c^{jk}$}; \draw[very thick] (0.3,1.4) -- (0.6,2.8);
\draw[very thick] (0,0) -- (1,1) node[right] {$a^{ki},c^{ki}$}; \draw[very thick] (1,1) -- (2,2);
\draw[very thick] (0,0) -- (1.3,0.4) node[below] {$a^{ij},c^{ij}$}; \draw[very thick] (1.3,0.4) -- (2.6,0.8);

\draw[fill=white] (0.6,0.8) circle(2pt) node[left] {$x_j$};
\draw[fill=black] (2.6,0.8) circle(2pt) node[below right] {$x_{ij}$};
\draw[fill=black] (0.6,2.8) circle(2pt) node[above left] {$x_{jk}$};

\draw[fill=black] (0,0) circle(2pt) node[below left] {$x_0$};
\draw[fill=black] (2,2) circle(2pt) node[below right] {$x_{ki}$};
\draw[fill=white] (2,0) circle(2pt) node[below right] {$x_i$};
\draw[fill=white] (0,2) circle(2pt) node[above left] {$x_k$};

\draw (3.5,1.3) node{$\mapsto$};

\begin{scope}[xshift=4.5cm,yshift=0]	
\draw (2,2) -- (0,2); 
\draw (2,2) -- (2,0); 

\draw (2.6,.8) -- (2,0); 
\draw (2,2) -- (2.6,2.8); 
\draw (0,2) -- (0.6,2.8); 

	\draw (0.6,0.8) -- (2,0.8);  \draw[dashed] (2,0.8) -- (2.6,0.8); 
	\draw (0.6,0.8) -- (0.6,2); \draw[dashed] (0.6,2) -- (0.6,2.8); 
	\draw (2.6,2.8) -- (0.6,2.8); 
	\draw (2.6,2.8) -- (2.6,0.8); 
	
\draw[very thick] (0.6,2.8) -- (1.3,2.4) node[right] {$\;\;\, a^{ij}_k,c^{ij}_k$}; \draw[very thick] (1.3,2.4) -- (2,2);
\draw[very thick] (2.6,0.8) -- (2.3,1.4) node[right] {$a^{jk}_i,c^{jk}_i$}; \draw[very thick] (2.3,1.4) -- (2,2);
\draw[very thick] (0.6,2.8) -- (1.6,1.8) node[below left] {$a^{ki}_j,c^{ki}_j$}; \draw[very thick] (1.6,1.8) -- (2.6,0.8);

\begin{scope}[xshift=6mm,yshift=8mm]	
	\draw[fill=white] (0,0) circle(2pt) node[below left] {$x_j$};
	\draw[fill=white] (2,2) circle(2pt) node[below right] {$x_{ijk}$};
	\draw[fill=black] (2,0) circle(2pt) node[below right] {$x_{ij}$};
	\draw[fill=black] (0,2) circle(2pt) node[above left] {$x_{jk}$};
\end{scope}

\draw[fill=black] (2,2) circle(2pt) node[right] {$\;x_{ki}$};
\draw[fill=white] (2,0) circle(2pt) node[below right] {$x_i$};
\draw[fill=white] (0,2) circle(2pt) node[above left] {$x_k$};
\end{scope}

\end{tikzpicture}

\caption[quad-equation]{Two-field star-triangle map}
\label{fig: star-triangle}
\end{figure}


The inverse map $G:(a^{ij}_k,c^{ij}_k,a^{jk}_i,c^{jk}_i,a^{ki}_j,c^{ki}_j)\mapsto(a^{ij},c^{ij},a^{jk},c^{jk},a^{ki},c^{ki})$ (the ``triangle-star'' map) is no more rational. Indeed, as one easily computes, the $c$-part of the map $F$ can be reversed in a unique way:
\begin{equation} \label{eq: c}
c^{ij}=c^{jk}_i+c^{ki}_j+\frac{(a^{ij})^2}{c^{ij}+c^{jk}+c^{ki}}=c^{jk}_i+c^{ki}_j+\frac{a^{jk}_ia^{ki}_j}{a^{ij}_k},
\end{equation}
and then the $a$-part of the map $F$ is reversed up to a simultaneous change of sign of all three variables $a^{ij}$, $a^{jk}$, $a^{ki}$, namely,  from equations \eqref{eq: a tilde} and \eqref{eq: c} we find:
\begin{equation}
a^{ij}=\frac{D}{a^{ij}_k}, \quad a^{jk}=\frac{D}{a^{jk}_i}, \quad a^{ki}=\frac{D}{a^{ki}_j}, 
\end{equation}
where $D$ is one of the square roots of the quantity
$$
D^2=a^{ij}_ka^{jk}_ia^{ki}_j(c^{ij}+c^{jk}+c^{ki})=(a^{ij}_ka^{jk}_i)^2+(a^{jk}_ia^{ki}_j)^2+(a^{ki}_ja^{ij}_k)^2+2a^{ij}_ka^{jk}_ia^{ki}_j(c^{ij}_k+c^{jk}_i+c^{ki}_j).
$$
Within one cube, one can trade a change of sign of $D$ for changing the sign of $x_0$.

\textbf{Remark.} Under the reduction $a^{ij}=c^{ij}$ corresponding to the discrete complex analysis, both relations \eqref{eq: a tilde} and \eqref{eq: c tilde} reduce to the classical star-triangle relation, cf. \cite{BS_pluriharmonic}:
\begin{equation} \label{star-triangle}
a^{ij}_k = \frac{a^{jk}a^{ki}}{a^{ij}+a^{jk}+a^{ki}},
\end{equation}
while the inverse ``triangle-star'' map becomes rational:
\begin{equation} \label{triangle-star}
a^{ij} = \frac{a^{ij}_ka^{jk}_i+a^{jk}_ia^{ki}_j+a^{ki}_ja^{ij}_k}{a^{ij}_k}.
\end{equation}

\begin{theorem} \label{theorem 2-field star triangle}
The two-field star-triangle map \eqref{eq: a tilde}, \eqref{eq: c tilde} is 4D consistent.
\end{theorem}
\begin{proof}
Let us sketch the exact statements which have to be checked. Actually, the proof should be done separately depending on the color of the vertex $n$ of the 4D cube $\sigma^{1234}$.

\emph{Case 1: $n$ is a black vertex.} One first applies the star-triangle map $F$ within the four 3D cubes $\sigma^{ijk}$ adjacent to $n$: 
$$
(a^{ij}_k,c^{ij}_k,a^{jk}_i,c^{jk}_i,a^{ki}_j,c^{ki}_j)=F(a^{ij},c^{ij},a^{jk},c^{jk},a^{ki},c^{ki}),
$$
and then applies the triangle-star map $G$ within the four 3D cubes $T_\ell \sigma^{ijk}$ adjacent to the opposite vertex $n+e_1+e_2+e_3+e_4$:
$$
(T_ka^{ij}_\ell,T_ia^{jk}_\ell,T_ja^{ki}_\ell, T_kc^{ij}_\ell,T_ic^{jk}_\ell,T_jc^{ki}_\ell)=G(a^{ij}_\ell,a^{jk}_\ell,a^{ki}_\ell, c^{ij}_\ell,c^{jk}_\ell,c^{ki}_\ell).
$$
The 4D consistency means $T_ka^{ij}_\ell=T_{\ell}a^{ij}_k$ and $T_kc^{ij}_\ell=T_{\ell}c^{ij}_k$ for the results coming from the 3D cubes  $T_\ell \sigma^{ijk}$ and $T_k\sigma^{ij\ell}$. Due to the presence of square roots in the formulas for the $a$-part of the map $G$, we actually show that $(T_ka^{ij}_\ell)^2=(T_{\ell}a^{ij}_k)^2$.

\emph{Case 2: $n$ is a white vertex.} One first applies the triangle-star map $G$ within the four 3D cubes $\sigma^{ijk}$ adjacent to $n$: 
$$
(a^{ij}_k,c^{ij}_k,a^{jk}_i,c^{jk}_i,a^{ki}_j,c^{ki}_j)=G(a^{ij},c^{ij},a^{jk},c^{jk},a^{ki},c^{ki}),
$$
and then applies the star-triangle map $F$ within the four 3D cubes $T_\ell \sigma^{ijk}$ adjacent to the opposite vertex $n+e_1+e_2+e_3+e_4$:
$$
(T_ka^{ij}_\ell,T_ia^{jk}_\ell,T_ja^{ki}_\ell, T_kc^{ij}_\ell,T_ic^{jk}_\ell,T_jc^{ki}_\ell)=F(a^{ij}_\ell,a^{jk}_\ell,a^{ki}_\ell, c^{ij}_\ell,c^{jk}_\ell,c^{ki}_\ell).
$$
Again, 4D consistency means $T_ka^{ij}_\ell=T_{\ell}a^{ij}_k$ and $T_kc^{ij}_\ell=T_{\ell}c^{ij}_k$ for the results coming from the 3D cubes  $T_\ell \sigma^{ijk}$ and $T_k\sigma^{ij\ell}$. This is fulfilled for any choice of the square roots for the $a$-part of the map $G$ within the 3D cubes $\sigma^{ijk}$.

All these verifications are done by a straightforward computation, which is well suitable for a symbolic software, but can be easily done also with bare hands. 
\end{proof}

A relation to Dirichlet energy and Laplace equations of Section \ref{sect Laplace} is as follows. Any rhombically embedded quad-graph $G$ (at least, a quasicrystallic one, with a finite number $m$ of different rhombus sides) can be realized as a quad-surface $\Sigma_G$ in $\bbZ^m$, as described in \cite{BMS, BS_DDG_book}. The action functional $S_G=\sum_{\sigma^{ij}\in \Sigma_G}\mathcal L(\sigma^{ij})$ is the Dirichlet energy on $G$. The pluri-Lagrangian property ensures that 
for quad-graphs related by a star-triangle flip, as illustrated on Figure \ref{fig: flip quad-graph}, we have also a relation between the corresponding Laplace operators and between solutions of the Laplace equations on both graphs. 

\begin{figure}[htbp]
\centering
\begin{tikzpicture}[scale=1.3]

\draw (1,0) -- (3,0); 
\draw (3,0) -- (4.732,1); 
\draw (4.732,1)  -- (2.732,1); 
\draw (2.732,1) -- (1,0); 

\draw (1,0) -- (0.628,1.970);
\draw (0.628,1.970) -- (2.360,2.970);
\draw (2.360,2.970) -- (2.732,1);

\draw (2.360,2.970) -- (4.360,2.970);
\draw (4.360,2.970) -- (4.732,1);

\draw[fill=black] (2.732,1) circle(2pt) node[below right] {$\; x_0$};

\draw[fill=black] (3,0) circle(2pt) node[below right] {$x_{ij}$};
\draw[fill=black] (0.628,1.970) circle(2pt) node[above left] {$x_{ki}$};
\draw[fill=black] (4.360,2.970) circle(2pt) node[right] {$\; x_{jk}$};

\draw[fill=white] (1,0) circle(2pt) node[left] {$x_i$};
\draw[fill=white] (2.360,2.970) circle(2pt) node[above left] {$x_k$};
\draw[fill=white] (4.732,1) circle(2pt) node[below right] {$x_j$};

\draw[very thick] (2.732,1) -- (3,0); 
\draw[very thick] (2.732,1) -- (0.628,1.970); 
\draw[very thick] (2.732,1) -- (4.360,2.970); 

\draw[-{>[scale=1]}] (1.3,0) arc (0:30:.3);
\draw[-{>[scale=1]}] (4.68,1.295) arc (100:180:.3);
\draw[-{>[scale=1]}] (2.1,2.82) arc (210:280:.3);

\draw (1.35,0.15) node [right] {$\pi-\phi_1$};
\draw (4.5,1.25) node [left] {$\pi-\phi_2$};
\draw (2.55,2.55) node [below left] {$\pi-\phi_3$};

\draw (5.6,1.3) node{$\mapsto$};

\begin{scope}[xshift=5.7cm,yshift=0]	
\draw (1,0) -- (3,0); 
\draw (3,0) -- (4.732,1);

\draw (1,0) -- (0.628,1.970);
\draw (0.628,1.970) -- (2.360,2.970);

\draw (2.360,2.970) -- (4.360,2.970);
\draw (4.360,2.970) -- (4.732,1);

\draw (3,0)  --  (2.628,1.970); 
\draw  (0.628,1.970) --  (2.628,1.970);
\draw (4.360,2.970) --   (2.628,1.970); 

\draw[fill=white] (2.628,1.970) circle(2pt) node[below left] {$x_{ijk}$};

\draw[fill=black] (3,0) circle(2pt) node[below right] {$x_{ij}$};
\draw[fill=black] (0.628,1.970) circle(2pt) node[above left] {$x_{ki}$};
\draw[fill=black] (4.360,2.970) circle(2pt) node[right] {$\; x_{jk}$};

\draw[fill=white] (1,0) circle(2pt) node[left] {$x_i$};
\draw[fill=white] (2.360,2.970) circle(2pt) node[above left] {$x_k$};
\draw[fill=white] (4.732,1) circle(2pt) node[below right] {$x_j$};

\draw[-{>[scale=1]}] (1.3,0) arc (0:100:.3);
\draw[-{>[scale=1]}] (4.68,1.295) arc (100:210:.3);
\draw[-{>[scale=1]}] (2.1,2.82) arc (210:360:.3);

\draw (1.2,0.3) node [right] {$\phi_2$};
\draw (3.9,1.2) node [right] {$\phi_3$};
\draw (2.4,3.2) node [right] {$\phi_1$};

\draw[very thick]  (3,0) -- (0.628,1.970) -- (4.360,2.970) -- (3,0); 
\end{scope}

\end{tikzpicture}

\caption[quad-equation]{Star-triangle flip on a bipartite rhombic quad-graph (for better readability the angles at the white vertices of the rhombi are marked)}
\label{fig: flip quad-graph}
\end{figure}
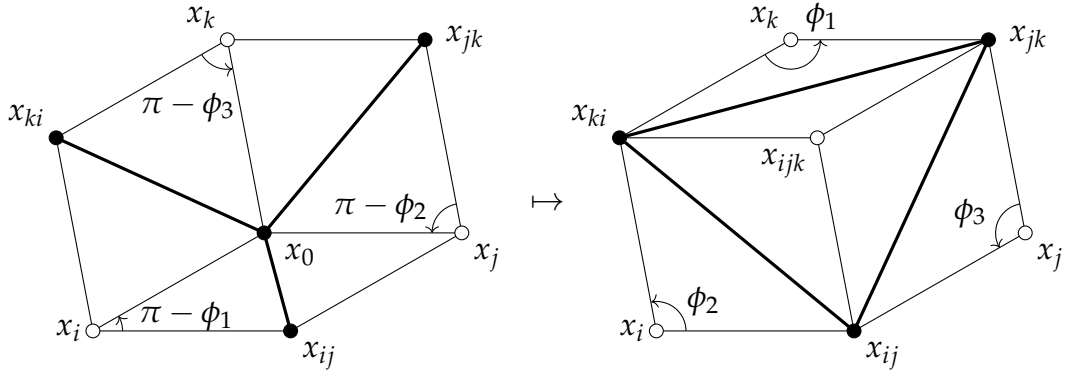


Dirichlet energy of of Section \ref{sect Laplace} corresponds to the following special solution of the two-field star-triangle map. If $\sigma^{ij}\in\Sigma_G$ is realized as a rhombus in $G$ with the angle $\phi$ at the black vertices, then
$$
a^{ij}=f(\phi)=\frac{\vartheta_1(\phi/2)}{\vartheta_2(\phi/2)}, \quad c^{ij}=g_0(\phi)=-\frac{1}{\vartheta_3\vartheta_4}\cdot\frac{\vartheta_2'(\phi/2)}{\vartheta_2(\phi/2)}.
$$
The star-triangle transformation acts on three rhombi in $G$ adjacent to a black point $x_0$, corresponding to $\sigma^{ij}$, $\sigma^{jk}$, $\sigma^{ki}$, with the  angles at the black vertices $\phi_1$, $\phi_2$, $\phi_3$ satisfying $$\phi_1+\phi_2+\phi_3=2\pi.$$ The images $T_k\sigma^{ij}$, $T_i\sigma^{jk}$, $T_j\sigma^{ki}$ are realized as three rhombi in the flipped graph, adjacent to a white point $x_{ijk}$, with the angles at the black vertices $\pi-\phi_1$, $\pi-\phi_2$, $\pi-\phi_3$. Thus,
$$
a^{ij}=f(\phi_1), \quad a^{jk}=f(\phi_2), \quad a^{ki}=f(\phi_3), 
$$
$$
c^{ij}=g_0(\phi_1), \quad c^{jk}=g_0(\phi_2), \quad c^{ki}=g_0(\phi_3),
$$
and
$$
a^{ij}_k=f(\pi-\phi_1), \quad a^{jk}_i=f(\pi-\phi_2), \quad a^{ki}_j=f(\pi-\phi_3), 
$$
$$
c^{ij}_k=g_0(\pi-\phi_1), \quad c^{jk}_i=g_0(\pi-\phi_2), \quad c^{ki}_j=g_0(\pi-\phi_3).
$$
In particular, for this solution we have $a^{ij}a_k^{ij}=a^{jk}a_i^{jk}=a^{ki}a_j^{ki}=1$.


\section{Conclusions}
\label{sect Conclusions}

Let us summarize the main findings of this paper. 

In the first part of the paper, we classify linear integrable (multi-dimensionally consistent) quad-equations \eqref{quad eq} on bipartite isoradial quad-graphs in $\bbC$, enjoying natural symmetries and the property that the restriction of their solutions to the black vertices satisfies the Laplace type equation \eqref{Laplace eq}. The classification reduces to solving a functional equation \eqref{eq fff} for the coefficient $f(\alpha,\beta)$ of equation \eqref{quad eq}. In the case when $f(\alpha,\beta)=f(\alpha-\beta)$, we give a complete solution of the functional equation, which is expressed in terms of elliptic functions. There are two real analytic reductions, corresponding to the cases when the underlying complex torus is of a rectangular type or of a rhombic type. The solution corresponding to the rectangular type was previously found in \cite{BTR, dT}. Using the multi-dimensional consistency, we construct the discrete exponential function on V(G) which serves as a basis of solutions of the quad-equation. 

There are indications that some of the results of \cite{BTR, dT} admit generalizations related to algebraic curves of higher order, see \cite{Fock, George}. It would be very important to find out which results of the present paper can be generalized in this direction.

In the second part of the paper, we focus on the integrability of discrete linear variational problems. We consider discrete pluri-harmonic functions, as defined in\cite{BS_pluriharmonic}, in the case when the underlying discrete 2-form depends on the fields at black vertices only. In an important particular case, we show that the problem reduces to a two-field generalization of the classical star-triangle map. We prove the integrability of this novel 3D system by showing its multi-dimensional consistency. The Laplacians from the first part come as a special solution of the two-field star-triangle map. We leave a general classification of discrete pluri-harmonic functions as an important open problem.


\section*{Acknowledgment}

This research is supported by the DFG Collaborative Research Center TRR 109 ``Discretization in Geometry and Dynamics''.


\begin{appendix}

\section{A pluri-Lagrangian generalization of the discrete 2-form \eqref{Dirichlet gg}}
\label{Appendix A}

In this section, it will be enough to make only local considerations, within one elementary cube $\sigma^{123}$, where we will assume for definiteness that the point $n$ is black. We set for the three cyclic permutations $(ijk)$ of $(123)$:
\begin{equation} \label{eq: 2-form bw corner 1}
\cL(\sigma^{ij})=\frac{1}{2}(b^{ij}x^2+c^{ij}x_{ij}^2)-a^{ij}xx_{ij},
\end{equation}
and
\begin{equation} \label{eq: 2-form bw corner 2}
\cL(\sigma^{ij}_k)=\frac{1}{2}(b^{ij}_k x_{jk}^2+c^{ij}_k x_{ki}^2)-a^{ij}_k x_{jk}x_{ki}.
\end{equation}
Thus, $a^{ij}$ are naturally assigned to the black diagonals of the plaquettes $\sigma^{ij}$, while $b^{ij}$, $c^{ij}$ are naturally assigned to the corresponding black corners of those plaquettes. Moreover, inspired by the formula \eqref{eq 2-form complete square}, we require the quadratic forms \eqref{eq: 2-form bw corner 1}, \eqref{eq: 2-form bw corner 2} to be complete squares:
\begin{equation}\label{cond complete square}
b^{ij}c^{ij}=(a^{ij})^2, \quad b^{ij}_kc^{ij}_k=(a^{ij}_k)^2.
\end{equation}
The exterior derivative $S^{123}=d\cL(\sigma^{123})$, as defined in \eqref{eq: Sijk}, is computed as follows:
\begin{equation}
S^{123}= \sum_{(ijk)={\rm c.p.}(123)}\left(\frac{1}{2}(b^{ij}_k x_{jk}^2+c^{ij}_k x_{ki}^2)-a^{ij}_k x_{jk}x_{ki}\right)-\left(\frac{1}{2}(b^{ij} x^2+c^{ij} x_{ij}^2)-a^{ij}xx_{ij}\right).
\end{equation}
Therefore, the system of corner equations consists of
\begin{eqnarray}
\frac{\partial S^{123}}{\partial x} &=& -(b^{12}+b^{23}+b^{31})x+a^{12}x_{12}+a^{23}x_{23}+a^{31}x_{31}=0,\\
\frac{\partial S^{123}}{\partial x_{ij}} &=& a^{ij}x+(c^{jk}_i+b^{ki}_j-c^{ij}) x_{ij}-a^{ki}_j x_{jk}-a^{jk}_i x_{ki}=0.
\end{eqnarray}
\begin{theorem}
Lagrangian 2-forms \eqref{eq: 2-form bw corner 1}, \eqref{eq: 2-form bw corner 2} with coefficients satisfying \eqref{cond complete square} are consistent, i.e. form a pluri-Lagrangian system, if and only if their coefficients satisfy
\begin{equation}\label{3-field s-t map a}
a^{12}_3=\frac{a^{23}a^{31}}{b^{12}+b^{23}+b^{31}}, \quad  a^{23}_1=\frac{a^{31}a^{12}}{b^{12}+b^{23}+b^{31}}, \quad  a^{31}_2=\frac{a^{12}a^{23}}{b^{12}+b^{23}+b^{31}}, 
\end{equation}
\begin{equation}\label{3-field s-t map b}
b^{12}_3=\frac{c^{23}b^{31}}{b^{12}+b^{23}+b^{31}}, \quad  b^{23}_1=\frac{c^{31}b^{12}}{b^{12}+b^{23}+b^{31}}, \quad b^{31}_2=\frac{c^{12}b^{23}}{b^{12}+b^{23}+b^{31}},
\end{equation}
\begin{equation}\label{3-field s-t map c}
c^{12}_3=\frac{b^{23}c^{31}}{b^{12}+b^{23}+b^{31}}, \quad  c^{23}_1=\frac{b^{31}c^{12}}{b^{12}+b^{23}+b^{31}}, \quad  c^{31}_2=\frac{b^{12}c^{23}}{b^{12}+b^{23}+b^{31}}.
\end{equation}
The rational map $\{a^{ij},b^{ij},c^{ij}\}\mapsto\{a^{ij}_k,b^{ij}_k,c^{ij}_k\}$ is a three-field analog of the star-triangle map, and reduces to the latter in the case $a^{ij}=b^{ij}=c^{ij}$. 
\end{theorem}
\begin{proof} The rank 1 conditions are equivalent to equations \eqref{3-field s-t map a} combined with
\begin{equation}
c^{jk}_i+b^{ki}_j-c^{ij}=-\frac{a^{jk}_ia^{ki}_j}{a^{ij}_k}=-\frac{(a^{ij})^2}{b^{12}+b^{23}+b^{31}}.
\end{equation}
so that, taking equation \eqref{cond complete square} into account, we come to
\begin{equation}
c^{jk}_i+b^{ki}_j=\frac{c^{ij}(b^{jk}+b^{ki})}{b^{12}+b^{23}+b^{31}}.
\end{equation}
These three equations together with (the second half of) \eqref{cond complete square} gives us a system of six equations for six variables $c^{ij}_k$, $b^{ij}_k$. Its solution is given by \eqref{3-field s-t map b}, \eqref{3-field s-t map c}.
\end{proof}
A serious drawback of the three-field analog of the star-triangle map is that it is not invertible. More precisely, the image domain of this map is described by the conditions
\begin{equation}\label{cond 3-field t-s}
 b^{12}_3 b^{23}_1 b^{31}_2=c^{12}_3 c^{23}_1 c^{31}_2
\end{equation}
(obviously satisfied by virtue of  \eqref{3-field s-t map b}, \eqref{3-field s-t map c}). On the five-dimensional submanifold of $\bbC^6$ described by \eqref{cond 3-field t-s}, relations \eqref{3-field s-t map b}, \eqref{3-field s-t map c} can be reversed up to the multiplicative scaling $b^{ij}\to \lambda^2 b^{ij}$, while relations \eqref{3-field s-t map a} can be reversed up to the multiplicative scaling $a^{ij}\to \lambda a^{ij}$.
 
 Thus, a 3D corner of a multi-dimensional quad-surface centered at a white point can only be flipped (according to the triangle-star map), if the coefficients on the black diagonals of this 3D corner satisfy conditions of the type \eqref{cond 3-field t-s}. It is very difficult to control such conditions in the process of flipping. The way out consists in interpreting these conditions, as well as relations
\begin{equation}\label{closed 2}
b^{ij}_k b^{jk} c^{ki}= c^{ij}_k c^{jk} b^{ki},
\end{equation} 
which are immediate consequences of  \eqref{3-field s-t map b}, \eqref{3-field s-t map c},  as closedness of the multiplicative 1-form $\nu$ on (directed) black diagonals, defined by
 \begin{equation}
\nu(n,n+e_i+e_j)=\frac{b^{ij}}{c^{ij}}, \quad  \nu(n+e_j+e_k,n+e_k+e_i)=\frac{b^{ij}_k}{c^{ij}_k}.
 \end{equation}
We consider only topologically trivial situations, where we can conclude that the multiplicative 1-form $\nu$ is exact. Let $\rho^2:\bbZ^m_{\rm even}\to\bbC$ be the function on black points obtained by integration of the exact 1-form $\nu$, i.e. satisfying
\begin{equation}
\nu(n,n+e_i+e_j) =\frac{\rho^2(n+e_i+e_j)}{\rho^2(n)}=\frac{\rho_{ij}^2}{\rho^2}, 
\end{equation}
\begin{equation}
\nu(n+e_j+e_k,n+e_k+e_i)  =\frac{\rho^2(n+e_k+e_i)}{\rho^2(n+e_j+e_k)}=\frac{\rho_{ki}^2}{\rho_{jk}^2}.
 \end{equation}
 Then we have:
 \begin{equation}\label{bc a}
 b^{ij}=a^{ij}\cdot\frac{\rho_{ij}}{\rho}, \quad c^{ij}=a^{ij}\cdot\frac{\rho}{\rho_{ij}},
 \end{equation}
 \begin{equation}\label{bc a hat}
 b^{ij}_k=a^{ij}_k\cdot\frac{\rho_{ki}}{\rho_{jk}}, \quad c^{ij}_k=a^{ij}_k\cdot\frac{\rho_{jk}}{\rho_{ki}}.
 \end{equation}
Conversely, for any function $\rho:\bbZ^m_{\rm even}\to\bbC$ equations \eqref{bc a}, \eqref{bc a hat} ensure that \eqref{cond 3-field t-s}, \eqref{closed 2} are satisfied. Now observe that substitution of \eqref{bc a}, \eqref{bc a hat} into \eqref{3-field s-t map a}, \eqref{3-field s-t map b}, \eqref{3-field s-t map c}, results in the same equations
\begin{equation}
a^{ij}_k=\frac{a^{jk}a^{ki}\rho}{a^{12}\rho_{12}+a^{23}\rho_{23}+a^{31}\rho_{31}}.
\end{equation}
The latter equation is equivalent to
\begin{equation}
a^{ij}_k\rho_{jk}\rho_{ki}=\frac{a^{jk}a^{ki}\rho^2\rho_{jk}\rho_{ki}}{a^{12}\rho\rho_{12}+a^{23}\rho\rho_{23}+a^{31}\rho\rho_{31}},
\end{equation}
which shows that the quantities
\begin{equation}
A^{ij}=a^{ij} \rho \rho_{ij}, \quad A^{ij}_k=a^{ij}_k \rho_{jk} \rho_{ki}
\end{equation}
satisfy the standard star-triangle relations 
\begin{equation}\label{star-triangle proper}
A^{ij}_k=\frac{A^{jk}A^{ki}}{A^{12}+A^{23}+A^{31}}.
\end{equation}
We get the following expressions:
\begin{equation}\label{3-field s-t map a sol 1}
a^{ij}=\frac{A^{ij}}{\rho\rho_{ij}}, \quad  b^{ij}=\frac{A^{ij}}{\rho^2}, \quad  c^{ij}=\frac{A^{ij}}{\rho_{ij}^2},
\end{equation}
\begin{equation}\label{3-field s-t map a sol 2}
a^{ij}_k=\frac{A^{ij}_k}{\rho_{jk}\rho_{ki}}, \quad  b^{ij}_k=\frac{A^{ij}_k}{\rho_{jk}^2}, \quad  c^{ij}_k=\frac{A^{ij}_k}{\rho_{ki}^2},
\end{equation}
which can be interpreted as a simple classification result.
\begin{theorem}
Lagrangian 2-forms \eqref{eq: 2-form bw corner 1}, \eqref{eq: 2-form bw corner 2} with coefficients satisfying \eqref{cond complete square} constitute a pluri-Lagrangian system, if and only if they are gauge equivalent to the 2-forms 
$$
 \frac{1}{2}A^{ij}(y_{ij}-y)^2, \quad \frac{1}{2}A^{ij}(y_i-y_j)^2,
$$ 
where $A^{ij}$ is an arbitrary solution of the star-triangle map \eqref{star-triangle proper}, via a gauge transformation $y=x/\rho$ with an arbitrary function $\rho:\mathbb Z^m_{\rm even}\to\mathbb C$.
\end{theorem}

\end{appendix}



\begin{thebibliography}{}

\bibitem{BTR} 
C. Boutillier, B. de Tili\`ere, and K. Raschel. The Z-invariant massive Laplacian on isoradial graphs. 
\emph{Invent. Math.}, 2017, \textbf{208}, Nr. 1, 109--189. 

\bibitem{BS_quad}
A.I. Bobenko, Yu.B. Suris. Integrable systems on quad-graphs. 
\emph{Intern. Math. Research Notices}, 2002, \textbf{2002}, Nr. 11, 573--611.

\bibitem{BS_DDG_book}
A.I. Bobenko, Yu.B. Suris. \emph{Discrete Differential Geometry. Integrable Structure}. 
Graduate Studies in Mathematics, Vol. 98. AMS, 2008. xxiv+404 pp. 

\bibitem{BS_pluriharmonic}
A.I Bobenko, Yu.B. Suris. Discrete pluriharmonic functions as solutions of linear pluri-Lagrangian systems. 
\emph{Commun. Math. Phys.}, 2015, \textbf{336}, Nr. 1, 199--215.

\bibitem{BMS} 
A. I. Bobenko, C. Mercat, and Y. B. Suris. Linear and nonlinear theories of discrete analytic functions. Integrable structure and isomonodromic Green’s function. \emph{J. Reine Angew. Math.}, 2005, \textbf{583}, 117--161.

\bibitem{BPS2}
R. Boll, M. Petrera, Yu.B. Suris. What is integrability of discrete variational systems? 
\emph{Proc. Royal Soc. A}, 2014, \textbf{470}, No. 2162, 20130550, 15 pp.

\bibitem{Chelkak_Smirnov}
D. Chelkak, S. Smirnov. Discrete complex analysis on isoradial graphs.
 \emph{Adv. Math.}, 2011, \textbf{228}, Nr. 3, 1590--1630.

\bibitem{Cooper}
S. Cooper. A functional equation and Jacobian elliptic functions. 
\emph{Aequationes Math.}, 1998, \textbf{56}, 69--80.

\bibitem{Fock}
V. Fock. Inverse spectral problem for GK integrable systems. 
{\tt arXiv:1503.00289 [math.AG]}.

\bibitem{George}
T. George. Spectra of biperiodic planar networks. 
{\tt arXiv:1901.06353 [math.CO]}.

\bibitem{Kenyon}
R. Kenyon. The Laplacian and Dirac operators on critical planar graphs. 
\emph{Invent. Math.}, 2002, \textbf{150}, Nr. 2, 409--439.

\bibitem{Krichever}
I. M. Krichever. Integrable linear equations and the Riemann–Schottky problem. In: \emph{Algebraic
geometry and number theory}, Progr. Math., 253, Birkh\"auser, Boston, 2006, 497--514.

\bibitem{Shiota}
T. Shiota, Characterization of Jacobian varieties in terms of soliton equations. 
\emph{Invent. Math.}, 1986, \textbf{83}, Nr. 2, 333--382.

\bibitem{dT}
B. de Tili\`ere. The Z-Dirac and massive Laplacian operators in the Z-invariant Ising model.
{\tt arXiv:1801.00207 [math-ph]}.

\bibitem{WW}
E.T. Whittaker, G.N. Watson. \emph{A Course of Modern Analysis}. Fourth Ed., Cambridge Univ. Press, 1927.




\end{thebibliography}
\end{document}